\documentclass[aps,prx,onecolumn,showpacs,superscriptaddress,longbibliography,floatfix]{revtex4-2}

\usepackage{amssymb, amsmath, bbm, amsthm, changes, graphicx, dsfont, verbatim, appendix, physics} 

\usepackage[hyperindex, breaklinks]{hyperref}
\hypersetup{
     colorlinks=true,        
     citecolor=blue,    
     filecolor=blue,      		
     urlcolor=blue,           	
    runcolor=cyan,
}

\usepackage{enumerate}
\usepackage{xcolor}
\usepackage{orcidlink}

\def\Z2{\mathbb Z_2}

\def\b|{\big \|}
\def\dist{\textrm{dist}}
\def\supp{\textrm{supp}}
\def\vlr{v_{LR}}
\def\dim{\textrm{dim}}

\theoremstyle{definition}
\newtheorem{Lemma}{Lemma}

\newtheorem{Theorem}[Lemma]{Theorem}

\newtheorem{Assumption}[Lemma]{Assumption}

\begin{document}
\title{The dynamical $\alpha$-Rényi entropies of local Hamiltonians grow at most linearly in time}

\author{Daniele Toniolo\,\orcidlink{0000-0003-2517-0770}}
\email[]{d.toniolo@ucl.ac.uk}
\email[]{danielet@alumni.ntnu.no}
\affiliation{Department of Physics and Astronomy, University College London, Gower Street, London WC1E 6BT, United
Kingdom}

\author{Sougato Bose\,\orcidlink{0000-0001-8726-0566}}
\affiliation{Department of Physics and Astronomy, University College London, Gower Street, London WC1E 6BT, United
Kingdom}


\begin{abstract}
We consider a generic one dimensional spin system of length $ L $, arbitrarily large, with strictly local interactions, for example nearest neighbor, and prove that the dynamical $ \alpha $-Rényi entropies, $ 0 < \alpha \le 1 $, of an initial product state grow at most linearly in time. This result arises from a general relation among dynamical $ \alpha $-Rényi entropies and Lieb-Robinson bounds. We extend our bound on the dynamical generation of entropy to systems with exponential decay of interactions, for values of $\alpha$ close enough to $ 1 $, and moreover to initial pure states with low entanglement, of order $ \log L $, that are typically represented by critical states. We establish that low entanglement states have an efficient MPS representation that persists at least up to times of order $ \log L $. 
The main technical tools are the Lieb-Robinson bounds, to locally approximate the dynamics of the spin chain, a strict upper bound of Audenaert on $ \alpha $-Rényi entropies and a bound on their concavity. Such a bound, that we provide in an appendix, can be of independent interest.
\end{abstract}

\maketitle



\section{Introduction}

$ \alpha $-Rényi entropies are generalized quantum entropies that are good quantifiers of entanglement \cite{Bennett_1996,Vedral_1998,Hayden_2006,Peschel_2011}. A list of their properties can be found in \cite{Vershynina_Scholarpedia}. Early works on the characterization of spin systems' ground state with Rényi entropies are \cite{Jin_2004,Franchini_2007}.

Experimental breakthroughs \cite{Jurcevic_2014,Islam_2015,Lukin_2019} have shown how quantum entropies can be detected and measured making them not only crucial quantities to characterize the behavior of many body systems from a theoretical point of view but also susceptible of experimental tests.

The scaling of $ \alpha $-Rényi entropies, $0 < \alpha < 1 $, of the reduced density matrix of a state, of a one-dimensional system, have been shown to determine whether such a state could be efficiently represented by a matrix product state (MPS) \cite{Verstraete_Cirac_2006, Schuch_2008}. In higher dimensional systems, instead, a non volume law for the Rényi entropy does not imply an efficient MPS representation \cite{Ge_Eisert_2016}.

The variation of the von Neumann entropy of the reduced density of a state evolved in time according to the unitary evolution of a system has attracted a lot of attention starting from the SIE (small incremental entangling) and SIM (small incremental mixing) conjectures attributed by Bravyi in \cite{Bravyi_2007} to Kitaev. See \cite{Bravyi_2007} also for early references on this topic. A recent review covering these themes is \cite{Gong_2022}. These conjectures roughly say that the rate of increase in time of the von Neumann entropy and the entropy of mixing are upper bounded by a constant independent from the overall system's size. Audenaert \cite{Audenaert_2014}, Van Acoleyen {\it et al.} \cite{Van_Acoleyen_2013} and Mariën {\it et al.} \cite{Marien_2016} have given proofs of such conjectures. In particular as a consequence of the (at most) linear growth in time of the entanglement entropy the authors of \cite{Van_Acoleyen_2013, Marien_2016} were able to prove, using the formalism of the quasi-adiabatic continuation \cite{Hastings_LSM_2004, Osborne_2007, Hastings_Locality_2010, Bachmann_Auto_2012, Nach_2019}, that states within the same phase of matter have the same scaling of the entanglement entropy, implying the stability of the area law in dimension larger than one. A phase of matter in this context is identified by the set of gapped eigenstates of local Hamiltonians that can be smoothly connected. Recently a concept of quantum phase has emerged in the more general context of Lindbladian evolution \cite{Coser_2019, Onorati_2023}. It is well known that the area law for one-dimensional gapped state has already been proven directly \cite{Hastings_2007,Arad_2012}. Recently proofs have also been given for two-dimensional frustration-free systems \cite{Anshu_2022_1,Anshu_2022_2}. In general a proof for higher dimensions is still lacking, nevertheless under additional mild conditions Masanes has proven a low entropy law in \cite{Masanes_2009}. Area laws have also been shown to follow from exponential decay of corrections \cite{Brandao_2013,Brandao_2014}. A rigorous analysis of entanglement rates for $\alpha$-Rényi entropy with $ \alpha > 1$ has been given in \cite{Shi_2022}, moreover for these entropies the authors of \cite{Bertini_2019} have shown the generation of entanglement, from a specific class of initial product states for the dynamics of the "self-dual" kicked Ising chains, exactly.

Calabrese and Cardy pioneered the use of field theoretic methods for the evaluation of entanglement entropies and their dynamics in spin systems  \cite{Calabrese_Cardy_2004,Calabrese_Cardy_2005,Calabrese_Cardy_2009}, providing in \cite{Calabrese_Cardy_2005} also the solution for the transverse Ising model. Other approaches in the context of integrable systems are \cite{Fagotti_2008} and, for $\alpha$-Rényi entropy with $ \alpha > 1$, \cite{Bertini_2022}.

Symmetries can affect the dynamical behavior of entropies, the authors of \cite{Rakovszky_2019,Huang_2020,Znidaric_2023} have pointed out sub-ballistic growth for diffusive systems.

In our work we are able to upper bound, as a function of time and $ \alpha $, the variation of the $\alpha$-Rényi entropy, with $0 < \alpha \le 1$,  of the partial trace of the time evolution of a generic product state, when the time evolution is induced by a nearest neighbor Hamiltonian $ H = \sum_{j=-L}^{L-1} H_{j,j+1} $ in one dimension. This corresponds to a global quench where the ground state of the initial Hamiltonian is a product state. With $r$ the local Hilbert space dimension, the bound reads:
\begin{equation} \label{entropy_bound}
 \Delta S_{\alpha}(t) \le K \, r^{\frac{1}{\alpha}-1} \, \max_j\{\|H_{j,j+1}\|\} \, t \, \log r + K'
\end{equation}
with both $ K $ and $ K' $ constants of order $ 1 $.
This is the first rigorous result, as far as we know, on the dynamical evolution of $\alpha$-Rényi entropy, with $0<\alpha < 1$. The upper bound on the von Neumann entropy follows as the  $\alpha \rightarrow 1$ limit.

The upper bound \eqref{entropy_bound} is obtained from the following general relation among $\alpha$-Rényi entropies, with $0<\alpha \le 1$ and Lieb-Robinson bounds. This is the most important result of our paper together with \eqref{entropy_bound} that is an application of theorem \ref{theorem_entropy_LR}. A further application of it has been obtained for the case of local Hamiltonian systems with slow dynamics (many-body localized for example), as quantified by a Lieb-Robinson bound giving rise to a logarithmic lightcone, in our work \cite{Toniolo_2024_3}, where we have proven a $ \log t $ behavior of $ \Delta S_\alpha (t) $. 

\begin{Theorem} (Entanglement generation, informal). \label{theorem_entropy_LR} 
For a one-dimensional lattice spin system, with local Hilbert space dimension equal to $ r $, with a Hamiltonian as in equation \eqref{Ham}, and Lieb-Robinson bound as quantified by $ \Delta_k(t) $ in \eqref{Delta_intro}, the generation of entanglement as quantified by the $ \alpha$-Rényi entropy, $0<\alpha \le 1$, starting from a product state $ \rho $, being $ l \ge 0 $ a variational parameter upon which \eqref{bound_sum_intro} must be minimized given the explicit expression of $ \Delta_k(t) $,  is upper bounded by: 
\begin{align} \label{bound_sum_intro}
 & \Delta S_\alpha (t) :=   S_\alpha \left( \Tr_{[1,L]} \rho(t) \right) - S_\alpha \left( \Tr_{[1,L]} \rho \right)  \nonumber \\ 
 & \le \frac{1}{1-\alpha} \sum_{k=l+1}^{L} \log \Big [ 1- \alpha \int_0^{t}ds \, \Delta_{k-1}(s)  + (r^{k+1}-1)^{1-\alpha} \left(\int_0^{t}ds \, \Delta_{k-1}(s)\right)^\alpha \Big ] + (l + 1)\log r 
\end{align}
$ \Delta_k(t) $ is upper bounded by:
\begin{align} \label{Delta_intro}
 & \Delta_k(t) \le \int_0^t ds \big \| \big [ H_{k,k+1}+H_{-k-1,-k}, e^{is(H_{\Lambda_k}-H_{[0,1]})}H_{[0,1]}e^{-is(H_{\Lambda_k}-H_{[0,1]})} \big ] \big \| 
\end{align}
$ H_{\Lambda_j}:=\sum_{k=-j}^{j-1} H_{k,k+1} $.

\end{Theorem}

\begin{proof}
 The reader will be guided from equation \eqref{Ham} to equation \eqref{from_52} to the proof of \eqref{bound_sum_intro}.
\end{proof}

When the shape of the Lieb-Robinson bound follows uniquely from the locality of the Hamiltonian, as explained for example in appendix \ref{L-R_bounds}, equation \eqref{bound_sum_intro} implies \eqref{entropy_bound}. A more restrictive Lieb-Robinson bound that implies, for example, a logarithm light cone, gives rise from \eqref{bound_sum_intro} to a dynamical generation of entropy that at long times follows a $ \log t $-law, as we have proven in \cite{Toniolo_2024_3}.

We are also able to generalize \eqref{entropy_bound}, see section \ref{low_ent_in_state}, to the evolution of an initial pure state with $\alpha$-Rényi entropy up to $ O \left( \log L \right) $, in this case the constant $ K' $ becomes of the order of the initial entropy of the state. Independently from the existence of an energy gap, the bound \eqref{entropy_bound} shows that a state of a one dimensional local system with interactions decaying fast enough, that has an efficient MPS representation, continues upon time-evolution to have an efficient MPS representation up to times at least of the order of $ \log L $.

The linear dependence on $ t $ of the upper bound on  $ \Delta S_\alpha(t) $ that we have established is the best possible for large $ t $. In fact in \cite{Schuch_Wolf_2008} a lower bound on the dynamical generation, starting from a product state, of the entanglement entropy for the Ising model was proven: $ \Delta S(t) \ge \frac{4}{3\pi} t -\frac{1}{2} \ln t -1 $.
The Ising model falls into the class of models, $ 1 $-dimensional nearest-neighbor, that we are considering here. $ \Delta S(t) $ denotes the variation of the von Neumann entropy.  $ \alpha $-Rényi entropies, with $0< \alpha < 1 $, are upper bounds of the von Neumann entropy, therefore they cannot growth slower than $ t $ for large $ t $.

The bound \eqref{entropy_bound} is easy to generalize to the case of a $ k $-neighbors Hamiltonian. We explicitly consider the case of interactions decreasing exponentially fast in section \ref{Renyi_bound_exp_tails}, proving with equation \eqref{Renyi_ent_bound_exp_tails} an upper bound linear in $ t $ as \eqref{entropy_bound} but only for values of $ \alpha $ sufficiently close, from below, to $ 1 $, according to \eqref{alpha_close_1}. 


It should be stressed that for pure states $ \rho = | \Phi \rangle \langle \Phi | $, $ | \Phi \rangle \in \mathcal{H}_1 \otimes \mathcal{H}_2 $, the $ \alpha = 1/2 $ Rényi entropy (of the partial trace) coincides with the logarithmic negativity $E_\mathcal{N}(\rho)$ \cite{Vidal_Werner}: 
\begin{align} \label{Renyi-logneg}
 E_\mathcal{N}(\rho) := \log_2 \| \rho^{T_1} \|_1 = S_{1/2}(\Tr_{\mathcal{H}_1} \rho)
\end{align}
In the appendix \ref{negativity_renyi} we show this equality following \cite{Vidal_Werner,Nach_2013}.
The logarithmic negativity has been studied in the context of many-body physics with field theoretical methods \cite{Calabrese_2012,Calabrese_2013,Alba_2013}, including its dynamics \cite{Coser_2014,Bertini_2022_1}.

Our work is organised as follows: in the section \ref{sec_Physical_Setting} we present the physical setting and discuss how the unitary dynamics is approximated making use of Lieb-Robinson bounds. In here we follow the approach employed in \cite{Osborne_2006, Eisert_Osborne_2006}. A short proof of the Lieb-Robinson bound, following \cite{Osborne_slides}, is in the appendix \ref{L-R_bounds}. In the section \ref{sec_Aud_bound} we present the upper bound on $\alpha$-Rényi entropies, $0 < \alpha < 1$, proven by Audenaert in \cite{Audenaert_2007} and discuss how it reduces to the more familiar Fannes-Audenaert-Petz bound \cite{Audenaert_2007, Winter_tight_2016} on the von Neumann entropy in the limit $\alpha \rightarrow 1$. In the section \ref{sec_upper_bound_local} we set up and sum up the series that gives rise to the upper bound on the Rényi entropies \eqref{entropy_bound}, discussing how the faster than exponential decrease of the spatial part of the Lieb-Robinson bound of a strictly local Hamiltonian enables the summation of the series for all $ 0 <  \alpha \le 1$, for arbitrarily large one-dimensional systems. We eventually state the final bound on the dynamical variation of the Rényi entropies \eqref{entropy_bound}. In the section \ref{Renyi_bound_exp_tails} we extend our bound to systems with exponential decay of interactions for values of $ 0 <  \alpha \le 1$ close enough to $ 1 $. We then show in \ref{low_ent_in_state} the generalization of the theory to pure states with entanglement of order $ \log L $. In the section \ref{Discussion} we discuss  possible extensions of our work and relations with other results in the literature. A set of appendices close the paper collecting proofs and technicalities.

\section{Physical Setting} \label{sec_Physical_Setting}

The physical setting is that of a one-dimensional spin chain with sites in the interval $ [-L,L] $. The local Hilbert space is $ \mathds{C}^r $. The Hamiltonian is a sum of nearest neighbors terms:
\begin{equation} \label{Ham}
  H=\sum_{j=-L}^{L-1} H_{j,j+1}
\end{equation}
$H_{j,j+1}$ is a short hand for $ \mathds{1}_r \otimes ... \otimes H_{j,j+1} \otimes ...  \otimes  \mathds{1}_r $, to make clear  that the support of $ H_{j,j+1} $ is on the sites $ \{j,j+1\} $.  We define $ J := \max_j \{ \|  H_{j,j+1} \| \} < \infty $, with $ \| \cdot \| $ denoting the operatorial norm, that is the maximum singular value.
Our final goal is the evaluation of the variation, associated to time evolution, of the $ \alpha $-Rényi entropies, with $ 0 < \alpha \le 1 $, of a reduced density matrix. Namely we want to upper bound as a function of time:
\begin{equation} \label{Delta_S}
\Delta S_\alpha (t) :=  | S_\alpha \left( \Tr_{[1,L]} \rho(t) \right) - S_\alpha \left( \Tr_{[1,L]} \rho \right) |
\end{equation} 
$ \Tr_{[1,L]} $ denotes partial tracing on the Hilbert space $ \bigotimes_{j=1}^L \mathcal{H}_j $, with $ \forall j $, $ \mathcal{H}_j = \mathds{C}^r $,  while $ \rho(t) $ is the evolution, at time $ t $ of the state $ \rho $, $ \rho(t) := e^{-itH} \rho e^{itH} $.
We rewrite the operator of unitary evolution $ e^{-itH} $, following the approach of \cite{Osborne_2006, Eisert_Osborne_2006}, to put in evidence that the term $ H_{0,1} $ of the Hamiltonian \eqref{Ham} is responsible for the ``spread'' of entanglement among the two halves of the system. We define
\begin{equation} \label{V_op}
 V(t):= e^{it (H_{[-L,0]}+H_{[1,L]})} e^{-itH} = e^{it (H-H_{[0,1]})} e^{-itH}
\end{equation}
$ H_{[-L,0]} := \sum_{j=-L}^{-1} H_{j,j+1} $ and $ H_{[1,L]} := \sum_{j=1}^{L-1} H_{j,j+1} $ collect the terms in the Hamiltonian \eqref{Ham} with support contained respectively in the intervals $ [-L,0] $ and $ [1,L] $. The term $ H_{0,1} $ has instead support on both these intervals.

Our theory holds with $ H $ as in \eqref{Ham} with open boundary conditions. If periodic boundary conditions were adopted instead two terms of the Hamiltonian, say $ H_{0,1} $ and $ H_{L,-L} $, would have overlapping supports with the two halves of the system after a bi-partition.

It is easy to verify that replacing $ e^{-itH} $ with $ V(t) $ leaves $ \Delta S_\alpha (t) $ in equation \eqref{Delta_S} invariant. In fact:
\begin{align}
& S_\alpha \big( \Tr_{[1,L]} V(t) \rho V^*(t) \big) = S_\alpha  \big( \Tr_{[1,L]} e^{it (H_{[-L,0]}+H_{[1,L]})} e^{-itH} \rho  e^{itH} e^{-it (H_{[-L,0]}+H_{[1,L]})}  \big)   \\
& = S_\alpha \big( \Tr_{[1,L]} e^{it H_{[-L,0]}} e^{itH_{[1,L]}} e^{-itH} \rho e^{itH} e^{-itH_{[1,L]}} e^{-it H_{[-L,0]}}  \big) \label{V10}   \\
&=S_\alpha \big[ e^{it H_{[-L,0]}} \big(  \Tr_{[1,L]}  e^{itH_{[1,L]}} e^{-itH} \rho e^{itH} e^{-itH_{[1,L]}}  \big) e^{-it H_{[-L,0]}} \big] \label{V11}  \\
&=S_\alpha \big(  \Tr_{[1,L]}  e^{itH_{[1,L]}} e^{-itH} \rho  e^{itH} e^{-itH_{[1,L]}} \big) \label{V12}  \\
&=S_\alpha \big(  \Tr_{[1,L]} e^{-itH} \rho e^{itH} \big) \label{V13} 
\end{align}
In \eqref{V10} we have used the fact that $ H_{[-L,0]} $ and $ H_{[1,L]} $ have disjoint supports threfore they commute, in \eqref{V11} that $ H_{[-L,0]} $ is unaffected by the partial trace because is not supported on the Hilbert space where the partial trace is acting upon, in \eqref{V12} that the Rényi entropy of a state is invariant by unitary conjugation of that state, finally in \eqref{V13} we have used the fact that $ H_{[1,L]} $ has support contained in the Hilbert space that is traced out therefore the cyclic property of the trace applies.


%
%

The operator $ V(t) $ in equation \eqref{V_op} also equals:
\begin{align} \label{int_picture}
V(t)= T \Big[ \exp \big(-i \int_0^t ds & e^{is(H_{[-L,0]}+H_{[1,L]})} \, H_{[0,1]} \, e^{-is(H_{[-L,0]}+H_{[1,L]})}  \big) \Big]
\end{align}
This is the so called interaction picture. $ T(\cdot) $ denotes time ordering, according to the rule: the  latest time goes on the left.
Equation \eqref{int_picture} is proven showing that, as \eqref{V_op}, is a solution of the differential equation
\begin{align}
  i \frac{d}{dt}V(t) = e^{it(H-H_{[0,1]})} H_{[0,1]} e^{-it(H-H_{[0,1]})} V(t) 
\end{align}
with a unique solution (once the initial condition $V(0)=\mathds{1}$ is set), therefore they coincide. The interaction picture \eqref{int_picture} already provides the hint that the support of $ V(t) $ expands  starting from the origin of the system, $ x=0 $, following the evolution of $ H_{[0,1]} $, that, captured by the Lieb-Robinson bound of the system, is at most linear in time.

\section{Audenaert bound on $\alpha $-Rényi entropies, $0 < \alpha < 1 $ } \label{sec_Aud_bound}

In the paper \cite{Audenaert_2007} Audenaert presents a sharp form of the Fannes bound on the difference among the von Neumann entropy of a pair of states (density matrices) $ \rho $ and $ \sigma $. The same bound also appears in the book of Petz \cite{Petz_Quantum_Information_book}, theorem 3.8, where the proof is attributed to Csiszár.

As a byproduct, Audenaert was able to obtain in \cite{Audenaert_2007} a sharp bound on the difference among the $ \alpha $-Rényi entropies, with $ 0 < \alpha < 1 $, of $ \rho $ and $ \sigma $. In what follows all the logarithms are in base $2$. The $ \alpha $-Rényi entropies, with $ \alpha \neq 1 $, are defined as:
\begin{equation} \label{Renyi_def}
 S_\alpha (\rho) = \frac{1}{1-\alpha} \log \Tr \rho^\alpha
\end{equation}
It holds $ S_\alpha (\cdot) \le S_\beta (\cdot) $ with $ \alpha \ge \beta $. This means that our attention is focused on quantum entropies that upper bound the von Neumann entropy, $ S $, that is the $ \alpha \rightarrow 1 $ limit of the $ \alpha $-Rényi entropies. 
\begin{equation}
  S (\rho) =  -\Tr \rho \log \rho 
\end{equation}
Denoting $ \lambda_j $ the eigenvalues of $ \rho $ in \eqref{Renyi_def}, this fact can be easily seen.
\begin{align}
&S_\alpha (\rho) = \frac{1}{1-\alpha} \log \sum_j \lambda_j^\alpha = \frac{1}{1-\alpha} \Big(\log \sum_j \lambda_j^\alpha |_{\alpha=1} + (\frac{\partial}{\partial \alpha} \log \sum_j \lambda_j^\alpha) |_{\alpha=1}(\alpha-1) + o(\alpha-1)  \Big)
\end{align}
With $ (\frac{\partial}{\partial \alpha} \log \sum_j \lambda_j^\alpha) |_{\alpha=1} = \sum_j \lambda_j \log \lambda_j $ and taking $ \lim_{\alpha \rightarrow 1} $ the result follows. 
With  $ T := \frac{1}{2}\|\rho-\sigma\|_1 $ the trace distance among the two states $ \rho $ and $ \sigma $, and $ d $ the dimension of the Hilbert space where the states are acting upon, the bound of Audenaert reads as equation A3 of \cite{Audenaert_2007}: 
\begin{equation} \label{Audenaert_bound}
 |S_\alpha(\rho)-S_\alpha(\sigma)| \le \frac{1}{1-\alpha} \log \left[ \left( 1-T \right)^\alpha + (d-1)^{1-\alpha}T^\alpha \right]
\end{equation}
This inequality is sharp, in the sense that for every value of $ T $, $ 0 \le T \le 1 $, there is a pair of states, at distance $ T $, that saturates the bound. These are given by:
\begin{equation} \label{saturating_states} 
 \rho=\textrm{Diag}(1,0, ... , 0), \hspace{1mm} \sigma=\textrm{Diag} \left( 1-T,\frac{T}{d-1}, ... ,\frac{T}{d-1} \right)
\end{equation}
The inverse is also true, namely that these are the only states that saturates the bound, see reference  \cite{Hanson_2017_1}.
We see that with $ \alpha = 0 $ the upper bound \eqref{Audenaert_bound} becomes $ T $ independent and therefore trivial, being always equal to $ \log d $, that is the maximal value of the $ \alpha$-Rényi entropy  $ \forall \alpha \ge 0 $. The $ \alpha = 0 $ Rényi entropy is also called Hartley, or max, entropy \cite{Vershynina_Scholarpedia}. In what follows we explicitly assume $ 0 < \alpha \le 1 $, with  $ \alpha = 1 $ taken as a limit.  
The states \eqref{saturating_states} play the same role also in the Audenaert-Fannes-Petz upper bound on the difference of von Neumann entropy:
\begin{equation} \label{Fannes_Audenaert_bound}
 |S(\rho)-S(\sigma)| \le T \log (d-1) + H_2(T,1-T)
\end{equation}
$ H_2(T,1-T) $ denotes the binary Shannon entropy: $ H_2(T,1-T) := -T\log T -(1-T) \log (1-T) $.
It can be easily shown that the bound \eqref{Fannes_Audenaert_bound} is the limit for $ \alpha \rightarrow 1 $ of the bound  \eqref{Audenaert_bound}. The Rényi entropies, with the von Neumann entropy being the continuation at $ \alpha = 1 $, are decreasing in $ \alpha $ for all positive $ \alpha $. This can be shown, for example, using the concavity of the logarithm and the Jensen inequality. The upper bound \eqref{Audenaert_bound} is also decreasing with $ 0 \le \alpha < 1$, we explicitly show this in the appendix \ref{decrease_Audenaert_bound}.

A comparison among the bounds \eqref{Audenaert_bound} and \eqref{Fannes_Audenaert_bound} 
shows that in the first one the trace distance among states is couple with the
exponential of the number of qudits in the systems (that is proportional to the dimension
of the Hilbert space), instead in the latter one is coupled with the number of qudits. This
means that it is harder to obtain an upper bound, for the dynamically generated
entanglement, that is system-size independent when the $ 0 \le \alpha < 1$ Rényi entropies are concerned.

Both the bounds \eqref{Audenaert_bound} and \eqref{Fannes_Audenaert_bound} are increasing in $ T $ only with $ T \in [0,1-1/d] $. With $ T = 1- 1/d $ the maximum value is reached, that for both entropies is equal to $ \log d $. As a consequence, since it is usually only possible to obtain upper bounds on the trace  distance the following slight modification of the bound \eqref{Audenaert_bound}, that is given in theorem 3.1 of \cite{Hanson_2017_2}, is useful. With $ T \le R \le 1 $ it is:

\begin{equation} \label{Audenaert_Datta}
  |S_\alpha(\rho)-S_\alpha(\sigma)| \le 
\begin{cases}
\frac{1}{1-\alpha} \log ( (1-R)^\alpha + (d-1)^{1-\alpha} R^\alpha), \hspace{1mm} \textrm{with} \, R \le 1-\frac{1}{d} \\
\log d, \hspace{1mm} \textrm{with} \, R \ge 1 - \frac{1}{d}
\end{cases}
\end{equation}
In \eqref{Audenaert_Datta} $ \Delta S_\alpha $ is increasing in $ R $. 
For the analogous resetting of the bound given in \eqref{Fannes_Audenaert_bound}, we refer to lemma 1 of \cite{Winter_tight_2016} and \cite{Hanson_2017_2}. Recently a new upper bound on the variation of von Neumann entropy that makes use of the trace distance and the norm distance among states has appeared \cite{Jabbour_2023}.

A further upper bound to \eqref{Audenaert_bound} and \eqref{Audenaert_Datta}, that is not tight but increasing in $ R $, for all $ 0\le R \le 1$, is: 
\begin{equation} \label{Audenaert_bound_sim}
  |S_\alpha(\rho)-S_\alpha(\sigma)| \le 
\frac{1}{1-\alpha} \log ( 1- \alpha R + (d-1)^{1-\alpha} R^\alpha) 
\end{equation}
This follows from $ (1-u)^\alpha \le (1-\alpha u) $ for $ u \le 1 $ and $ 0 < \alpha < 1$, and the fact that the logarithm is an increasing function.

An important remark about the bounds \eqref{Audenaert_bound} and \eqref{Fannes_Audenaert_bound} is that their derivatives in the trace distance, when evaluated at $ T=0 $, diverge. It means that increasing $ T $, from $ T=0 $, where the two states coincide, the difference in entropy increases at a diverging rate. For the von Neumann entropy and the $ \alpha$-Rényi entropy, with $ \alpha > 1 $, it has been established \cite{Van_Acoleyen_2013, Marien_2016, Vershynina_2019} that, for local Hamiltonians systems, the dynamically generated entropy actually grows at a finite rate from $ t=0 $.

In this work we will be interested in the dynamical generation of entropy over any time interval, proving that with strictly local Hamiltonians is always independent form the system's size for all $ \alpha  \in (0,1) $ , while when interactions decay exponentially there always exist, for any finite system, a set of $ \alpha < 1 $ close to $ 1 $ such that the scaling is system-size independent.

\section{Upper bounding $ \Delta S_\alpha (t) $ for a nearest neighbor Hamiltonian} \label{sec_upper_bound_local}

In this section we make use of the following notation: $ \Lambda_j := [-j,j] $, $ V_{\Lambda_j} $ is the $ V$-evolution associated to the Hamiltonian  $ H_{[-j,j]}:=\sum_{k=-j}^{j-1} H_{k,k+1} $, namely:
\begin{align} \label{V_restricted}
 & V_{\Lambda_j}(t):= e^{it (H_{[-j,0]}+H_{[1,j]})}  e^{-itH_{[-j,j]}} = e^{it (H_{[-j,j]}-H_{[0,1]})} e^{-itH_{[-j,j]}}
\end{align}
According to \eqref{V_restricted} it is: $ V_{\Lambda_L}(t):=V(t) $.

We have already defined in \eqref{Delta_S}, $ \Delta S_\alpha (t) :=  | S_\alpha \left( \Tr_{[1,L]} \rho(t) \right) - S_\alpha \left( \Tr_{[1,L]} \rho \right) | $ as the object of our investigation. We now set an upper bound to  $ \Delta S_\alpha (t) $ in the form of a telescopic sum:
\begin{align}
& \Delta S_\alpha (t)  =  | S_\alpha \left( \Tr_{[1,L]} V(t) \rho V^*(t) \right) - S_\alpha \left( \Tr_{[1,L]} \rho \right) | \nonumber \\
 & = | S_\alpha \left( \Tr_{[1,L]} V_{\Lambda_L}(t) \rho V_{\Lambda_L}^*(t) \right)  - S_\alpha \left( \Tr_{[1,L]}   V_{\Lambda_{L-1}}(t) \rho V_{\Lambda_{L-1}}^*(t) \right) + S_\alpha \left( \Tr_{[1,L]}   V_{\Lambda_{L-1}}(t) \rho V_{\Lambda_{L-1}}^*(t) \right) \nonumber \\
&-  ... + S_\alpha \left( \Tr_{[1,L]}   V_{\Lambda_{l}}(t) \rho V_{\Lambda_{l}}^*(t) \right) - S_\alpha \left( \Tr_{[1,L]}   \rho \right)  | \label{tele} \\
 & \le \sum_{k=l+1}^{L} | S_\alpha \left( \Tr_{[1,L]}   V_{\Lambda_{k}}(t) \rho V_{\Lambda_{k}}^*(t) \right) - S_\alpha \left( \Tr_{[1,L]}   V_{\Lambda_{k-1}}(t) \rho V_{\Lambda_{k-1}}^*(t) \right)| +  | S_\alpha \left( \Tr_{[1,L]}  V_{\Lambda_{l}}(t) \rho V_{\Lambda_{l}}^*(t) \right) - S_\alpha \left( \Tr_{[1,L]}   \rho \right)  | \label{upper_bound_series}
\end{align}
At this point a crucial observation  follows from the assumption of $ \rho $ being a product state
\begin{equation}
 \rho = \bigotimes_{j=-L}^L \rho_j 
\end{equation}
with $ \rho_j : \mathds{C}^r \rightarrow \mathds{C}^r $. This  implies that each $ \Tr_{[1,L]}   V_{\Lambda_k}(t) \rho V_{\Lambda_k}^*(t) $ factorizes: the unitary $ V_{\Lambda_k} $ is supported on the interval $ [-k,k] $, with a slight abuse of notation we write $ V_{\Lambda_k} = V_{\Lambda_k} \otimes \mathds{1}_{[-k,k]^c}$, where $ [-k,k]^c $ denotes the complement of $ [-k,k] $ in $ [-L,L] $. Moreover with 
\begin{equation}
 \rho_{[-k,k]}:= \bigotimes_{j=-k}^k \rho_j 
\end{equation}
we have that:
\begin{align} \label{product_state}
& \Tr_{[1,L]} \left( V_{\Lambda_k} \rho V^*_{\Lambda_k} \right) = \Tr_{[1,k]}  \Tr_{[k+1,L]}   \left( V_{\Lambda_k} \rho_{[-k,k]} V^*_{\Lambda_k} \otimes \rho_{[-k,k]^c} \right) = \Tr_{[1,k]}     \left( V_{\Lambda_k} \rho_{[-k,k]} V^*_{\Lambda_k} \right) \otimes \Tr_{[k+1,L]} \rho_{[-k,k]^c} 
\end{align}
It follows that:
\begin{align}
 & | S_\alpha \left( \Tr_{[1,L]}   V_{\Lambda_{k}}(t) \rho V_{\Lambda_{k}}^*(t) \right) - S_\alpha \left( \Tr_{[1,L]}   V_{\Lambda_{k-1}}(t) \rho V_{\Lambda_{k-1}}^*(t) \right)| \\
 & = 
 | S_\alpha \left(\Tr_{[1,k]}     \left( V_{\Lambda_k} \rho_{[-k,k]} V^*_{\Lambda_k} \right) \otimes \Tr_{[k+1,L]} \rho_{[-k,k]^c} \right) - S_\alpha \left( \Tr_{[1,k]}     \left( V_{\Lambda_{k-1}} \rho_{[-k,k]} V^*_{\Lambda_{k-1}} \right) \otimes \Tr_{[k+1,L]} \rho_{[-k,k]^c} \right) | \label{emb} \\ 
 & = | S_\alpha \left( \Tr_{[1,k]} \left( V_{\Lambda_k} \rho_{[-k,k]} V^*_{\Lambda_k} \right) \right) - S_\alpha \left( \Tr_{[1,k]}  \left( V_{\Lambda_{k-1}} \rho_{[-k,k]} V^*_{\Lambda_{k-1}} \right) \right) | \label{S_prod}
\end{align}
In \eqref{emb} we have embedded $ V_{\Lambda_{k-1}} $ in the Hilbert space over the interval $ [-k,k] $ with a tensorization with the identity, that means (with a slight abuse of notation): $ V_{\Lambda_{k-1}} = \mathds{1}_{r}\otimes V_{\Lambda_{k-1}} \otimes \mathds{1}_{r}$. In equation \eqref{S_prod} we have used the additivity of Rényi entropies: $ S_\alpha(\rho \otimes \sigma) = S_\alpha(\rho) + S_\alpha(\sigma) $. We then realize that the density matrices in equation \eqref{S_prod} are supported on the Hilbert space over the interval $ [-k,0] $ that has dimensionality $ r^{k+1} $. 
We define:
\begin{align} \label{del_rho}
 T_k(t):= & \frac{1}{2} \b|  \Tr_{[1,k]}  \left( V_{\Lambda_{k}} \rho_{[-k,k]} V^*_{\Lambda_{k}} \right) - \Tr_{[1,k]}  \left( V_{\Lambda_{k-1}} \rho_{[-k,k]} V^*_{\Lambda_{k-1}} \right)   \b|_1 
\end{align}
Using the upper bound \eqref{Audenaert_bound_sim} for each term of the sum \eqref{upper_bound_series} and taking into account \eqref{S_prod} we have that:
\begin{align} \label{bound_sum_1}
  \Delta S_\alpha (t) \le & \frac{1}{1-\alpha} \sum_{k=l+1}^{L} \log \Big[ 1-\alpha T_k(t)  + (r^{k+1}-1)^{1-\alpha}T_k(t)^\alpha \Big] + (l + 1)\log r 
\end{align}
We have upper bounded the last term in \eqref{upper_bound_series}, $ | S_\alpha \left( \Tr_{[1,L]}   \rho_{l}(t) \right) - S_\alpha \left( \Tr_{[1,L]}   \rho \right)  |$, with the trivial bound for density matrices supported on a Hilbert space of dimension $ r^{l+1} $, that is $ (l + 1)\log r  $. This choice will be justified a posteriori after the minimization of \eqref{upper_bound_series} with respect to $ l $.

The relation among the trace distance $ T_k(t) $ and $ \Delta_{k-1}(t) $, as given in equation \eqref{C12}, namely
\begin{align} \label{from_52}
 T_k(t) \le \int_0^{|t|} ds \Delta_{k-1}(s) 
\end{align}
inserted into \eqref{bound_sum_1} allows to obtain equation \eqref{bound_sum_intro} of theorem \ref{theorem_entropy_LR}.

We now proceed to prove equation \eqref{entropy_bound} taking into account the bound $ \Delta_{k}(t') \le \frac{(4Jt')^k}{k!} $ arising from a generic nearest neighbor one-dimensional spin Hamiltonian \eqref{Ham}, as shown in the appendix \ref{L-R_bounds}.

The upper bound on the trace distance that is developed in the appendix \ref{bound_trace_dist}, being $ J := \max_j \{ \|H_{j,j+1} \| \} $ and the rescaled time $ t':= 4 \, J \, t $, reads:
\begin{align} \label{trace_dist}
 T_k(t) & \le \frac{1}{2} \|  V_{\Lambda_{k}} \rho_{[-k,k]} V^*_{\Lambda_{k}}  -V_{\Lambda_{k-1}} \rho_{[-k,k]} V^*_{\Lambda_{k-1}} \|_1  \le \frac{1}{4} \frac{(4 J t)^{k}}{k!} = \frac{1}{4} \frac{t'^{k}}{k!} 
\end{align}
We have used the fact that for a generic matrix $ A $ defined on a bipartite Hilbert space $ \mathcal{H}_1 \otimes \mathcal{H}_2 $, it holds $ \| \Tr_{\mathcal{H}_j} A \|_1 \le \|  A \|_1 $, with $ j \in \{1,2\} $, as a reference see, for example, equations 7 and 17 of \cite{Rastegin_2012}.
It also follows from: 
\begin{equation}
 \frac{1}{\dim(\mathcal{H}_1)} \mathds{1}_{\mathcal{H}_1} \otimes \left( \Tr_{\mathcal{H}_1} A \right)  = \int_W  \, W^* A  W d_{Haar}W 
\end{equation}
The integral above is performed according to the Haar measure over all the unitaries $ W $ supported on $ \mathcal{H}_1 $.
There is a generalization to all CPTP maps $  \mathcal{N} $, namely $ \| \mathcal{N}(\rho) -\mathcal{N}(\sigma) \|_1 \le \| \rho - \sigma \|_1 $.

The faster than exponential decrease in $ k $ of the upper bound \eqref{trace_dist} on $ T_k(t) $, ensures that, when coupled in \eqref{bound_sum_1} with  $ (r^{k+1}-1)^{1-\alpha} $,  the sum in \eqref{bound_sum_1} will converge. We will consider the $ L \rightarrow \infty $ limit.
Inserting \eqref{trace_dist} into \eqref{bound_sum_1}, and defining $ u_k(t'):= \frac{1}{4} \left( \frac{e t'}{k} \right)^k \ge \frac{1}{4} \frac{t'^{k}}{k!} $, see the brief discussion in appendix \ref{on_u} on $ u_k(t') $, we get:
\begin{align}\label{bound_sum_3}
 \Delta S_\alpha (t) & \le \frac{1}{1-\alpha} \sum_{k=l+1}^{L}  \log \Big [ 1- \alpha u_k(t') + r^{(k+1)(1-\alpha)} u_k(t')^\alpha \Big ] + (l + 1)\log r
\end{align}

In the appendix \ref{alg} is proven that \eqref{bound_sum_3}, at a fixed $t'$, has a unique minimum with respect to $ l $.  
A positive integer $ l $ that minimizes \eqref{bound_sum_3} satisfies the following pair of conditions. Denoting $ f(l,t') $ the RHS of \eqref{bound_sum_3}, the minimum is reached for that value of $ l $ such that $ f(l,t') $ is smaller of both $ f(l+1,t') $ and $ f(l-1,t') $.
\begin{equation}
 \begin{cases}
 &f(l,t') \le f(l+1,t') \\
 &f(l,t')  \le f(l-1,t')
 \end{cases}
\end{equation}
These equations must be thought at a fixed $ \alpha $ and at a fixed time $ t' $, with $ l $ determined as a function of $ \alpha $ and $ t' $. Then the equations for the minimization of $ \Delta S_\alpha(t) $ are:
\begin{equation}
 \begin{cases} \label{system_min}
& \frac{1}{1-\alpha} \ \log \left[ 1- \alpha u_{l+1}(t') + r^{(l+2)(1-\alpha)} u_{l+1}(t')^\alpha \right] \le \log r  \\
 & \frac{1}{1-\alpha} \ \log \left[ 1- \alpha u_l(t') + r^{(l+1)(1-\alpha)} u_l(t')^\alpha \right] \ge \log r 
 \end{cases}
\end{equation}
 This implies that despite we are not able to solve explicitly the system \eqref{system_min}, that can be done numerically, the approximate solution for $ u_{l+1}(t') $ that we are going to provide simply gives rise to a further upper bound to \eqref{Delta_S}. Nevertheless, as discussed after \eqref{entropy_bound}, for a generic local (nearest neighbor in this case) Hamiltonian the linear dependence in $ t' $ of the upper bound to $  \Delta S_\alpha $ is the best possible.

We choose to provide an approximate solution to \eqref{system_min} solving, again approximately, with $ \alpha $ far enough from $ 1 $, the following equation
\begin{equation}  \label{approx_min}
 \frac{1}{1-\alpha} \ \log \left[ r^{(l+1)(1-\alpha)} u_{l}(t')^\alpha \right] \approx \log r
\end{equation}
That leads to:
\begin{align}
 & u_{l}(t') \approx r^{-l\frac{1-\alpha}{\alpha}} \\
 & \frac{t'^{l}}{4}\left(\frac{e}{l}\right)^{l} \approx r^{-l\frac{1-\alpha}{\alpha}} \\
 & \frac{t'}{4^{\frac{1}{l}}}\left(\frac{e}{l}\right) \approx r^{-\frac{1-\alpha}{\alpha}} \\
 & e \, r^{\frac{1-\alpha}{\alpha}} \, t' \approx l \, 4^\frac{1}{l}
\end{align}
It is: $ 1 < 4^\frac{1}{l} \le 4 $ with $ l \ge 1 $. As the approximate solution of \eqref{system_min} we pick: 
\begin{equation} \label{l_alpha}
 l_\alpha(t') := c \, r^{\frac{1-\alpha}{\alpha}} \, t' 
\end{equation}
with $ c > e $. The reason for $ c > e $ in our approach is to ensure exponential decrease in $ t $ of the sum in  \eqref{bound_sum_3}, as discussed in appendix \ref{sum_details}.
We stress again that despite \eqref{l_alpha} being an approximate solution of \eqref{system_min}, when $ l $  is replaced with $ l_\alpha(t') $ in \eqref{bound_sum_3}, this still provides an upper bound to $ \Delta S_{\alpha}(t) $ because of the uniqueness of the minimum in $ l $ of \eqref{bound_sum_3}. We remark that $ l $ is an integer, while in general $ l_\alpha(t') $ is not an integer. Nevertheless \eqref{bound_sum_3} is still well defined with $ l $ real, see the change of variable $ k \rightarrow n $ in equation \eqref{var_change}.

Plugging \eqref{l_alpha} into \eqref{bound_sum_3}, we expect that for all $ t' $ the sum in \eqref{bound_sum_3} will be at most of order $ 1 $ because of the fast decrease of the $ u_k(t') $, at a fixed $ t' $, with $ k \ge l_\alpha(t') +1 $. The upper bound to \eqref{bound_sum_3}  is performed, in a way that allows us to take the limit $ \alpha \rightarrow 1 $, in appendix \ref{sum_details}. Recollecting that $ t':=4Jt $, this leads to 
\begin{equation} \label{bound_delta_S}
 \Delta S_{\alpha}(t) \le K \, r^{\frac{1}{\alpha}-1} \, J \, t \, \log r + K'
\end{equation}
with $ K=4c $ and $ K' \le \left(1+\frac{e}{4c(e-1)}\right)\log r $, see appendix \ref{sum_details}, both of order $ 1 $. The bound \eqref{bound_delta_S} holds in the limit of a system of infinite size $ L \rightarrow \infty $, in this case the saturation time for the entropy, namely the time such that the entropy reaches its maximal value, that is the logarithm of the size of the Hilbert space, is infinite as well, therefore we can look at the asymptotic behavior at long times of \eqref{bound_delta_S}, that is:
\begin{equation} \label{asimp_bound_delta_S}
 \lim_{t' \rightarrow \infty} \lim_{L \rightarrow \infty} \frac{\Delta S_{\alpha}(t')}{t'} \le c \, r^{\frac{1}{\alpha}-1}  \, \log r 
\end{equation}

It is interesting to compare our upper bound with the entanglement entropy rate obtained in \cite{Marien_2016} in the case of the von Neumann entropy, $ \alpha = 1 $. Their equation (133), being an upper bound on the rate for all times, implies, in their notations:
\begin{align} \label{Verst_rate}
 \Delta S_1(t) \le 2^{\nu +1} c' (\log d) A \sum_r r^{2\nu} \|h(r)\| \, t
\end{align}
In a one dimensional system, $ \nu=1$, and for a strictly local Hamiltonian, in \eqref{Verst_rate} it is: $ A = 1 $, $ c'=2 $, $ \sum_r r^{2\nu} \|h(r)\| \le J $. For comparison, in our notations, it is: $  \Delta S_1(t) \le 8 (\log r) J \, t $. Looking at the coefficients of proportionality with $ t $ in \eqref{bound_delta_S}, with $ \alpha=1 $, that is $ 4cJ \log r $, it appears that, being $ c > e $, we do slightly worse than them. If instead we use the L-R bounds obtained from graph theory, see equation 3.51 of \cite{Chen_2023}, that allows to replace $ 4J $ to $ 2J $ in \eqref{asimp_bound_delta_S}, then in comparison with the results of \cite{Marien_2016}, we do slightly better.

As a final remark we recall that despite the maximal growth of entanglement entropy is linear in time for a generic product state, the authors of \cite{Hutter_2012} have established that for a system coupled to a large environment the von Neumann entropy rate of the system's reduced density matrix is low with high probability.

\section{Upper bounding  $ \Delta S_\alpha (t) $ for a system with exponential decay of interactions} \label{Renyi_bound_exp_tails}

We will extend the results about the dynamical generation of $ \alpha$-Rényi entropy of a strictly local Hamiltonian to the case of an Hamiltonian with exponentially decaying interactions.

We are able to generalize the approach that we set up in section \ref{sec_upper_bound_local} for nearest-neighbour Hamiltonians to Hamiltonians that are a sum of terms that are supported on the whole system but with decaying interactions. 
The scheme that we develop works for interactions with decay-rate that is at least exponential, $ f $ denotes such decay.

Denoting $ H = \sum_r H_r $ the Hamiltonian, each term $ H_r $ is supported on the whole lattice, we will refer to $ H_r $ as ``centered'' in $ r $. $ H_r $ is an interaction decaying according to the function $ f $, meaning that there exist constants $ J $ and $ \xi $, that are independent from $ r $, such that with $ r $ inside the region $ X $, that for simplicity is assumed connected, it holds:
\begin{equation} \label{decay}
\big \| \frac{1}{2^{|X^c|}} \left( \Tr_{X^c} H_r \right) \otimes \mathds{1}_{X^c} - H_r \big \| \le  J \,  f \left( \frac{\textrm{dist}(r,X^c)}{\xi} \right)
\end{equation}
For simplicity in \eqref{decay} we are considering the case of qubits, namely, the local Hilbert space dimension is equal to $ 2 $.
The physical interpretation of \eqref{decay} is straightforward: the approximation, in norm, of $ H_r $ with an operator that is supported on the region $ X $ containing $ r $, improves enlarging the distance of $ r $ from the complement of $ X $. The definition \eqref{decay} is employed, for example, in  \cite{Lu_2024}. We observe that if $ H_r $ was strictly local with support contained in $ X $, then the LHS of \eqref{decay} would be exactly vanishing. Also taking $ f $ with compact support the RHS of \eqref{decay}
vanishes with $ X $ large enough, this corresponds to Hamiltonians with $ k $ nearest neighbors, a generalization of the case studied in section \ref{sec_upper_bound_local}.

A Hamiltonian with decaying interactions does not provide an immediate way to identify the terms $ H_r $ responsible for the spread of entanglement across the two halves of the system, that was given for a generic nearest-neighbor Hamiltonian by $ H_{[0,1]} $. We introduce a distance $ b $, that will be determined by minimization of the upper bound to the dynamical Rényi entropy, such that, the terms of $ H $ mostly responsible for the connection among the left and right halves of the system are those centered within the distance $ b $ from the origin: $ \sum_{|r| < b} H_r $. The physical intuition is that $ b $ will be  in relation with the decaying length $ \xi $ of definition \eqref{decay}.

With this in mind we set up a different telescopic sum than \eqref{tele}. The first step is to rewrite:
\begin{align} \label{first}
  e^{-itH}  = e^{-it \left( \widehat{H}_{[-L,-b]} + \widehat{H}_{[b,L]} \right)} e^{it \left( \widehat{H}_{[-L,-b]} + \widehat{H}_{[b,L]} \right)} e^{-itH}
\end{align}
Where we have introduced, with $ 1 \le b \le j \le L $:
\begin{align} \label{1-2j}
 \widehat{H}_{[b,j]} := & \frac{1}{2^{|([1,2j] \cap [1,L])^c|}} \left( \Tr_{([1,2j] \cap [1,L])^c}  \sum_{r=b}^j H_r \right)  \otimes \mathds{1}_{([1,2j] \cap [1,L])^c} 
\end{align}
$ \widehat{H}_{[b,j]} $ is supported on $ [1,2j]  \cap [1,L] $. We also define: 
\begin{align} \label{-2j-2j}
 \widehat{H}_{[-j,j]} := & \frac{1}{2^{|([-2j,2j] \cap [-L,L])^c|}} \left( \Tr_{([-2j,2j]\cap [-L,L])^c}  \sum_{r=-j}^j H_r \right) \otimes \mathds{1}_{([-2j,2j]\cap [-L,L])^c} 
\end{align}


We notice that $ \widehat{H}_{[-L,L]} = H $. 
In equation \eqref{first} we have put in evidence at the exponent the sum of the terms centered on the left half and right half of the system outside the interval $ (-b,b) $ that, as we said, represents the region that we assume to ``connect'' the left and right halves and be responsible of the entanglement spread.

The first factor of \eqref{first} is not supported across the two halves of the system therefore it does not contribute to the entropy. We now look at the second and third factor in \eqref{first}. If we would assume that in our lattice we renamed the lattice site $ j=0 $ into $ j=-1 $, namely we count the lattice sites of our system only with strictly negative, or strictly positive integers (this trick was used for example in \cite{Lu_2024}), then if the Hamiltonian $ H $ was a nearest neighbor Hamiltonian, then the $ V $ operator that we introduced in \eqref{V_op}, following \cite{Eisert_Osborne_2006, Osborne_2006}, 
 would coincide with the product of second and third factor in \eqref{first}.

We define 
\begin{align}
  \widehat{V}_{\Lambda_j}(t) := e^{it \left( \widehat{H}_{[-j,-b]} + \widehat{H}_{[b,j]} \right)}   e^{-it \widehat{H}_{[-j,j]}}
\end{align}
that is supported on $ [-2j,2j] \, \cap \, [-L,L] $. Then:
\begin{align}
 & \Delta S_\alpha (t) \nonumber \\ 
 & =  | S_\alpha \left( \Tr_{[1,L]} e^{-itH} \rho e^{itH} \right) - S_\alpha \left( \Tr_{[1,L]} \rho \right) | \nonumber \\
 & = | S_\alpha \left( \Tr_{[1,L]}   \widehat{V}_{\Lambda_L}(t)  \rho  \widehat{V}_{\Lambda_L}^*(t)  \right) - S_\alpha \left( \Tr_{[1,L]}   \rho \right)  |  \\
 & = | S_\alpha \left( \Tr_{[1,L]}  \widehat{V}_{\Lambda_L}(t)  \rho  \widehat{V}_{\Lambda_L}^*(t)   \right) -  S_\alpha \left( \Tr_{[1,L]} \widehat{V}_{\Lambda_{L-1}}(t)   \rho  \widehat{V}_{\Lambda_{L-1}}^*(t) \right) +  ... + S_\alpha \left( \Tr_{[1,L]} \widehat{V}_{\Lambda_l}(t)  \rho  \widehat{V}_{\Lambda_l}^*(t) \right) - S_\alpha \left( \Tr_{[1,L]}   \rho \right) | \\ 
 & \le \sum_{k=l}^{L-1}  | S_\alpha \left( \Tr_{[1,L]}  \widehat{V}_{\Lambda_{k+1}}(t)  \rho  \widehat{V}_{\Lambda_{k+1}}^*(t)   \right) -  S_\alpha \left( \Tr_{[1,L]} \widehat{V}_{\Lambda_{k}}(t)   \rho  \widehat{V}_{\Lambda_{k}}^*(t) \right)| + | S_\alpha \left( \Tr_{[1,L]} \widehat{V}_{\Lambda_l}(t)  \rho  \widehat{V}_{\Lambda_l}^*(t) \right) - S_\alpha \left( \Tr_{[1,L]}   \rho \right) | 
\end{align}
At this point we again assume that $ \rho $ is a product state, therefore according to \eqref{product_state}-\eqref{S_prod}, we have:
\begin{align}
 &\Delta S_\alpha (t) \nonumber \\
 & \le \sum_{k=l}^{L-1}  | S_\alpha \left( \Tr_{[1,2(k+1)]}  \widehat{V}_{\Lambda_{k+1}}(t)  \rho_{[-2(k+1),2(k+1)]}  \widehat{V}_{\Lambda_{k+1}}^*(t)   \right) - S_\alpha \left( \Tr_{[1,2(k+1)]} \widehat{V}_{\Lambda_{k}}(t)   \rho_{[-2(k+1),2(k+1)]}  \widehat{V}_{\Lambda_{k}}^*(t) \right)|  \nonumber \\
 & + | S_\alpha \left( \Tr_{[1,2l]} \widehat{V}_{\Lambda_l}(t)  \rho_{[-2l,2l]}   \widehat{V}_{\Lambda_l}^*(t) \right) - S_\alpha \left( \Tr_{[1,L]}  \rho_{[-2l,2l]}  \right) | \label{44}
\end{align}
In the above we have used the fact that  $   \widehat{V}_{\Lambda_{k}}(t)  $ is supported on $ [-2k,2k] $. The dimensional reduction associated with product states allows the application of the Audenaert upper bound \eqref{Audenaert_bound}. We then need to provide an upper bound to the trace distance that according to appendix \ref{sub_bound_exp_tail}, equations \eqref{127} and \eqref{g_def}, reads:
\begin{align} \label{trace_dist_tails}
 &R_k(t)  := \frac{1}{2} \|  \Tr_{[1,2(k+1)]} \left( \widehat{V}_{\Lambda_{k+1}}(t)  \rho_{[-2(k+1),2(k+1)]}  \widehat{V}_{\Lambda_{k+1}}^*(t) \right) - \Tr_{[1,2(k+1)]} \left(  \widehat{V}_{\Lambda_{k}}(t)  \rho_{[-2(k+1),2(k+1)]}  \widehat{V}_{\Lambda_{k}}^*(t) \right) \|_1 \\
 & \le \frac{1}{2} \|  \widehat{V}_{\Lambda_{k+1}}(t)  \rho_{[-2(k+1),2(k+1)]}  \widehat{V}_{\Lambda_{k+1}}^*(t) - \widehat{V}_{\Lambda_{k}}(t)  \rho_{[-2(k+1),2(k+1)]} \widehat{V}_{\Lambda_{k}}^*(t) \|_1  \nonumber \\ &\le   t \, J \, O \left(  g(b,J,t) \, f \left( \frac{k}{4 \xi} \right) \right)  
\end{align}
With $ g(b,J,t):= b \,t \, J((e^{\vlr t}-1) + a) + c $, with $ a $ and $ c $ of $ O(1) $, and with $ \vlr $ the Lieb-Robinson velocity of the Hamiltonian. We can now apply to \eqref{44} the Audenaert bound \eqref{Audenaert_bound}, that reads:
\begin{align} \label{bound_sum_tails}
  \Delta S_\alpha (t) \le \frac{1}{1-\alpha} \sum_{k=l}^{L-1} \log \Big[ 1-\alpha R_k(t)   + (2^{2(k+1)+1}-1)^{1-\alpha}R_k(t)^\alpha \Big] + 2l +1
\end{align}
We now consider explicitly the case of an Hamiltonian that is a sum of terms with exponential tails $ f\left(\frac{k}{\xi} \right) = e^{-\frac{k}{\xi}} $. In this case the upper bound to  $ \Delta S_\alpha (t) $ is proven in appendix \ref{sub_bound_exp_tail} and given by \eqref{final_bound_tails_app} in the limit $ L\rightarrow \infty $. The final result, for $ \frac{1}{1+\frac{1}{8\xi\ln2}} < \alpha \le 1$  and with $ \beta > \frac{\vlr}{\frac{1}{4\xi}-(2\ln2)\frac{1-\alpha}{\alpha}} $, is:
\begin{align} \label{bound_Rényi_tails}
  \Delta S_\alpha (t) \le 2\beta t + A(t)
\end{align}
 with $ \lambda:= \frac{\alpha}{4\xi} - 2\ln2(1-\alpha) > 0 $, $ A(t) $ is decreasing exponentially fast in $ t $ and given in \eqref{final_bound_tails_app} as:
\begin{align}
& A(t) :=  -\frac{\alpha}{1-\alpha} \xi (t J)^2 e^{\vlr t}e^{-\frac{\beta t}{4\xi}} \frac{1}{e^{\frac{1}{4\xi}}-1} + \frac{1}{1-\alpha} e^{(3\ln 2)(1-\alpha)}\left(\xi (t J)^2 \right)^\alpha e^{\alpha \vlr t}e^{-\lambda \beta t} \frac{1}{e^\lambda -1} +1
\end{align}
With
\begin{align} \label{alpha_close_1}
  \frac{1}{1+\frac{1}{8\xi\ln2}} < \alpha \le 1
\end{align}
 this implies that: 
\begin{align} 
\lim_{t\rightarrow \infty} \lim_{L\rightarrow \infty} \frac{\Delta S_\alpha (t)}{t} \le 2\beta \nonumber
\end{align}
Considering the lower bound on $ \beta $, the best upper bound on $ \Delta S_\alpha (t) $ reads
\begin{align} 
\lim_{t\rightarrow \infty} \lim_{L\rightarrow \infty} \frac{\Delta S_\alpha (t)}{t}  \le  \frac{c' \vlr}{\frac{1}{4\xi}- (2\ln2)\frac{1-\alpha}{\alpha}}
\end{align}
with $ c' > 2 $.
The limit $ \alpha \rightarrow 1 $, that provides the von Neumann entropy, is discussed in \eqref{G66}, giving
\begin{align} \label{Renyi_ent_bound_exp_tails}
  & \lim_{t\rightarrow \infty}  \lim_{L\rightarrow \infty} \frac{\Delta S_1 (t)}{t} \le 4 c'\vlr \xi
\end{align}

\section{Dynamical $\alpha$-Rényi entropy for initial low entangled states} \label{low_ent_in_state}

In this section we extend the theory that have developed in section \ref{sec_upper_bound_local} for product states to initial low entangled states. This will allow us to show that, for one dimensional systems, states with an efficient matrix product state (MPS) representation continue to have such efficient representation, upon time evolution, at least up to times of the order of $ \log L $.

The $O(\log L) $ scaling for the entanglement entropy can be found for the ground state of critical systems, like the XY model \cite{Jin_2004,Peschel_2004,Latorre_2004,Franchini_2007}, the Ising and XXZ model \cite{Vidal_2003_1,Latorre_2004,Calabrese_Cardy_2004}, and the Hubbard model \cite{Korepin_2004}. Remarkably, Korepin in \cite{Korepin_2004} also shows, using conformal field theory methods and the results from \cite{Affleck_1986}, that the entanglement entropy is recovered from the thermodynamic entropy as the system's temperature goes to zero. On the other hand with positive temperature the scaling of the entropy with the system's size is a volume law. The $O(\log L) $ scaling also appears in free fermionic models with finite Fermi surface \cite{Wolf_2006,Gioev_2006}. For free fermions the $O(\log L)$ scaling has been proven for all $\alpha$-Rényi entropies, $ \alpha > 0 $, in \cite{Leschke_2014} using previous results of A. V.  Sobolev.

The main technical tools employed in this section are: an important result about the existence of an efficient MPS representation for low entangled states \cite{Verstraete_Cirac_2006,Schuch_2008} in one dimension and an upper bound on the concavity of $\alpha$-Rényi entropies, $0 <\alpha <1$, proven using the theory of majorization, in appendix \ref{Up_Renyi_Con}, this result can be of independent interest.

The results of this section are based on an assumption about the price paid, in terms of entropy, by replacing a product state with its time evolution, we clearly state this in \ref{Assum_1}.

Let us consider a low entangled state vector $ |v\rangle \in (\mathds{C}^2)^{\otimes (2L+1)} $, namely we assume that $ S_\alpha(\Tr_{[1,L]} |v\rangle \langle v| ) \le O(\log L) $. The results of \cite{Verstraete_Cirac_2006} ensure that a low entangled state $ |v\rangle $, together with an additional assumption on the distribution of the tails of the distribution of the Schmidt coefficients of $ |v\rangle \langle v| $, see after equation (4) of \cite{Verstraete_Cirac_2006}, has an efficient MPS representation. This means that the bond dimension (D) of the MPS is of order $ \textrm{poly}(L) $. In fact the rank of the state is upper bounded by $ D^2 $ \cite{Schuch_Wolf_2008, Scholl_2011}.  See, for example, \cite{Scholl_2011,Cirac_RMP_2021} for reviews on MPS and related topics. We express this fact saying that $ |v\rangle \langle v|  $  has rank of $O(\textrm{poly}(L)) $ or equivalently that is the linear combination of $ O(\textrm{poly}(L)) $ terms:
\begin{align} \label{linear_comb}
 &|v\rangle \langle v| := \underset{k_{-L},...,k_{L}}{\underset{j_{-L},...,j_{L}}{\sum}} c_{(j_{-L},...,j_{L})} \overline{c}_{(k_{-L},...,k_{L})}   |j_{-L} \rangle  \langle k_{-L}|  \otimes ... \otimes |j_{L} \rangle  \langle k_{L}|
\end{align}
The $ 2^{2L+1} $ states $ \{ |j_{-L} \rangle  \langle j_{-L}|  \otimes ... \otimes |j_{L} \rangle  \langle j_{L}| \} $ in \eqref{linear_comb} are all orthonormal.  


In section \ref{sec_upper_bound_local} we upper bounded $ \Delta S_\alpha(t) $ for initial product states. Restricting our attention to pure product states, it is  $ S_\alpha(0)=0 $. As show in appendix \ref{Up_Renyi_Con}, equation \eqref{convex_comb}, a convex combination of $ O(\textrm{poly}(L)) $ pure product states has dynamical $ \alpha$-Rényi entropy upper bounded by $ K_\alpha t + K' + O(\log L) $: 
\begin{align} \label{efficient_product}
 & S_\alpha \left( \Tr_{[1,L]} U(t) \left( \sum_j q_j | j \rangle \langle j | \right) U^*(t) \right) \le K_\alpha t + K' + O(\log L)
\end{align}
\eqref{efficient_product} implies that, for times $ t \le O(\log L) $, the state $U(t) \left( \sum_j q_j | j \rangle \langle j | \right) U^*(t) $ has an efficient matrix product operator (MPO) representation.

Let us consider now a state vector $ | v \rangle $  with low entanglement, as above, and make a further assumption on $ | v \rangle $ formalized as follows.

\begin{Assumption} \label{Assum_1}  The pure state $ |v\rangle \langle v| $, with $ |v\rangle \in (\mathds{C}^2)^{\otimes (2L+1)} $, satisfies:
\begin{enumerate}[(i)]
\item
$ S_\alpha(\Tr_{[1,L]} |v\rangle \langle v| ) \le O(\log L) $
\item
For $ t \le O(\log L) $, with $ U(t)=e^{-itH} $ and $ H $ given in \eqref{Ham}, $ |v\rangle $ further satisfies:
\begin{align}
 & S_\alpha \Big( \Tr_{[1,L]} \underset{j_{-L},...,j_{L}}{\sum} \left( U(t)  |j_{-L} \rangle  \langle j_{-L}|  \otimes ... \otimes |j_{L} \rangle  \langle j_{L}| U^*(t) \right) U(t)  | v \rangle \langle v | U^*(t) \left( U(t)  |j_{-L} \rangle  \langle j_{-L}|  \otimes ... \otimes |j_{L} \rangle  \langle j_{L}| U^*(t) \right) \Big) \nonumber \\ & \le 
   S_\alpha \Big( \Tr_{[1,L]} \underset{j_{-L},...,j_{L}}{\sum} \left(  |j_{-L} \rangle  \langle j_{-L}|  \otimes ... \otimes |j_{L} \rangle  \langle j_{L}|  \right) U(t) | v \rangle \langle v | U^*(t) \left( |j_{-L} \rangle  \langle j_{-L}|  \otimes ... \otimes |j_{L} \rangle  \langle j_{L}| \right) \Big) + O(\textrm{poly}(\log L)) \label{eq_assum}
\end{align}
 \end{enumerate} 

\end{Assumption}
Let us discuss the meaning of this assumption. In equation \eqref{eq_assum} we have replaced $ U(t)  |j_{-L} \rangle  \langle j_{-L}|  \otimes ... \otimes |j_{L} \rangle  \langle j_{L}| U^*(t) $ with $   |j_{-L} \rangle  \langle j_{-L}|  \otimes ... \otimes |j_{L} \rangle  \langle j_{L}|  $, considering this is a product state we know, because of our bound \eqref{bound_delta_S}, that its entropy increases at most linearly in time, so our reasoning is that overall within a time scale $ O(\log L) $ the error in entropy in the LHS of \eqref{eq_assum} would be of the same order. The upper bound of Audenaert, equation \eqref{Audenaert_bound}, that we have employed to prove \eqref{bound_delta_S}, turns out to be inefficient  to try to actually prove (ii) from (i) in assumption \ref{Assum_1}  because it requires a too small trace distance. At the same time Audenaert's bound is designed for two completely generic states while the pair of states that we are considering in \eqref{eq_assum} are in relation with each other.

\begin{Lemma} \label{lemma_MPS_time}
With $ | v \rangle $ as in assumption \ref{Assum_1}, it follows that:
\begin{align}
 & S_\alpha \left( \Tr_{[1,L]} U(t)  | v \rangle \langle v |  U^*(t) \right) \le   K_\alpha t + K' + O(\textrm{poly}(\log L))
\end{align}

\end{Lemma}

\begin{proof}

\begin{align}
 & S_\alpha \left( \Tr_{[1,L]} U(t)  | v \rangle \langle v |  U^*(t) \right)  \label{52} \\ 
 & \le S_\alpha \Big( \Tr_{[1,L]} \underset{j_{-L},...,j_{L}}{\sum} \left(  |j_{-L} \rangle  \langle j_{-L}|  \otimes ... \otimes |j_{L} \rangle  \langle j_{L}|  \right) U(t)  | v \rangle \langle v | U^*(t) \left( |j_{-L} \rangle  \langle j_{-L}|  \otimes ... \otimes |j_{L} \rangle  \langle j_{L}| \right) \Big) \label{53} \\
 & \le  S_\alpha \Big( \Tr_{[1,L]} \underset{j_{-L},...,j_{L}}{\sum}  U(t) \left(  |j_{-L} \rangle  \langle j_{-L}|  \otimes ... \otimes |j_{L} \rangle  \langle j_{L}| \right)   | v \rangle \langle v | \left( |j_{-L} \rangle  \langle j_{-L}|  \otimes ... \otimes |j_{L} \rangle  \langle j_{L}| \right) U^*(t)  \Big) \nonumber \\ & + O(\textrm{poly}(\log L)) \nonumber  \\
 & = S_\alpha \Big( \Tr_{[1,L]} \underset{j_{-L},...,j_{L}}  {\sum} |c_{(j_{-L},...,j_{L})}|^2 U(t) \left(  |j_{-L} \rangle  \langle j_{-L}|  \otimes ... \otimes |j_{L} \rangle  \langle j_{L}| \right) U^*(t)  \Big)  + O(\textrm{poly}(\log L)) \label{54}  \\
 & \le K_\alpha t + K' + O(\textrm{poly}(\log L)) \label{55} 
\end{align}
The inequality in \eqref{52} follows from pinching, see for example section II.5 of \cite{Bhatia}, and the fact that pinching with respect to a product basis and partial tracing commute. In \eqref{53} we made use of the assumption \ref{Assum_1}, and in \eqref{54} we employed the definition \eqref{linear_comb}, where following by the assumption  $ S_\alpha(\Tr_{[1,L]} |v\rangle \langle v| ) \le O(\log L) $, the number of non vanishing coefficients $ c_{(j_{-L},...,j_{L})} $ in \eqref{54} is at most of order $  O(\textrm{poly}(L)) $. Inequality \eqref{55} then follows from \eqref{convex_comb}.

\end{proof}

The conclusion from \eqref{55}, and lemma 2 of \cite{Verstraete_Cirac_2006}, is that a state $ |v\rangle $ with an efficient MPS representation will continue to have, under time evolution, such efficient representation for times $ t \le O(\log L) $.

%
%
%
%
%
%
%
%
%
%
%
%
%
%
%
%
%

\section{Outlook} \label{Discussion}

In this work we have made the first step towards a general theory that ensures to obtain from the Lieb-Robinson bounds of local Hamiltonians an upper bound to the dynamical generation of $ \alpha$-Rényi entropies. While here we developed the general theory and explicitly solved the case of systems with linear lightcones, in \cite{Toniolo_2024_3} we present an application to systems with logarithm lightcones, mostly in relation with the many-body-localization phenomenology.

A natural way to extend our theory is to look into higher dimensions and to consider explicitly the Hamiltonians with interactions decreasing slower than exponentially. This would allow to address the stability of the area law for $ \alpha$-Rényi entropies in local systems using the formalism of the quasi-adiabatic continuation, in analogy to what done for the von Neumann entropy by the authors of \cite{Van_Acoleyen_2013, Marien_2016}. Let us expand this remark.

Mariën {\it et al.} in \cite{Marien_2016} proved that, in any dimension, an area-law state remains area-law within the same phase of matter.
The definition of phase of matter that they refer to is as follows \cite{Hastings_Locality_2010}: let us consider a (local) Hamiltonian $H(s)$ depending smoothly on a parameter (or set of parameters) with a ground state state that remains gapped upon variation of the parameter. A quantum phase is the set of such ground states obtained varying the parameter.

In one dimension gapped ground states of local Hamiltonians are area law \cite{Hastings_2007}, therefore the stability of area law follows. This is true both for the von Neumann entropy and the $\alpha$-Rényi entropy, with $ 0 < \alpha <1 $, \cite{Huang_2014}.

In more than one dimension it is in general unknown if gapped ground states are area law, despite Anshu {\it et al.} proved a sub-volume law for 2d frustration-free spin systems \cite{Anshu_2022_1,Anshu_2022_2}.

 The mapping from one state to the other in the same phase is given by the so called quasi-adiabatic continuation: it is a unitary map generated by a many-body Hamiltonian, $ D(s) $, see below, with sub-exponential decaying terms, that in equation (2.19) of \cite{Bachmann_Auto_2012} reads:

\begin{align} \label{qa}
 D(s) = \int_{-\infty}^\infty dt W_\gamma(t) e^{itH(s)} \left( \frac{d}{ds}H(s) \right) e^{-itH(s)}
\end{align}
with $ W_\gamma(t) $ a filter function.

The entanglement rate theory of Mariën {\it et al.} applied to the quasi-adiabatic continuation shows that the difference among the von Neumann entropies of two states in the same phase is $ O(1) $ in any dimension.

\section*{Acknowledgments}
 Daniele Toniolo and Sougato Bose acknowledge support from UKRI grant EP/R029075/1.  The authors thank an anonymous referee for recalling equation \eqref{Renyi-logneg} to them.  D. T. is glad to acknowledge discussions about topics related to this work with Álvaro Alhambra, Ángela Capel, Angelo Lucia, Dylan Lewis, Emilio Onorati and Lluís Masanes.

\section*{Appendices}
\appendix

\section{Decrease in $ \alpha$ of the Audenaert upper bound on $ \Delta S_\alpha $} \label{decrease_Audenaert_bound}

$\alpha$-Rényi entropies converge to the von Neumann entropy both in the limit $ \alpha \uparrow 1 $ and $ \alpha \downarrow 1$. Therefore they are continued for  $ \alpha = 1 $ into the von Neumann entropy.
A simple application of the Jensen inequality in relation with the concavity of the  $ \log $ shows that with $ 0 <  \alpha < 1 $, $ S_\alpha $ is decreasing in $  \alpha $. In the same spirit we show that the Audenaert upper bound on $ \Delta S_\alpha $ of equation \eqref{Audenaert_bound}, and therefore \eqref{Audenaert_Datta}, is decreasing in $ \alpha $.

\begin{align}
 & \frac{d}{d \, \alpha}  \left[ \frac{1}{1-\alpha} \log \left( (1-T)^\alpha + (d-1)^{1-\alpha} T^\alpha \right) \right] = \nonumber \\ 
 & = \frac{1}{(1-\alpha)^2} \log \left( (1-T)^\alpha + (d-1)^{1-\alpha} T^\alpha \right)  \nonumber \\
 & + \frac{1}{1-\alpha} \frac{ (\log(1-T)) (1-T)^\alpha }{(1-T)^\alpha + (d-1)^{1-\alpha} T^\alpha } \nonumber \\
 & + \frac{1}{1-\alpha} \frac{ - (\log(d-1)) (d-1)^{1-\alpha} T^\alpha + (d-1)^{1-\alpha}  (\log T) T^\alpha }{(1-T)^\alpha + (d-1)^{1-\alpha} T^\alpha } \nonumber \\
 & = \frac{1}{(1-\alpha)^2} \log \left( (1-T)^\alpha + (d-1)^{1-\alpha} T^\alpha \right)  \nonumber \\
 & + \frac{1}{(1-\alpha)^2}  \frac{ (\log(1-T)^{1-\alpha}) (1-T)^\alpha }{(1-T)^\alpha + (d-1)^{1-\alpha} T^\alpha } \nonumber\\
 & + \frac{1}{(1-\alpha)^2}  \frac{  \log \left((d-1)^{\alpha-1} T^{1-\alpha} \right) (d-1)^{1-\alpha} T^\alpha }{(1-T)^\alpha + (d-1)^{1-\alpha} T^\alpha }\label{Jensen}
\end{align}
Defining 
\begin{align}
 & p_1 := \frac{(1-T)^\alpha}{(1-T)^\alpha + (d-1)^{1-\alpha} T^\alpha} \ge 0 \\
 & p_2 := \frac{(d-1)^{1-\alpha} T^\alpha}{(1-T)^\alpha + (d-1)^{1-\alpha} T^\alpha} \ge 0 \\
 & x_1 := (1-T)^{1-\alpha} \\
 & x_2 := (d-1)^{\alpha-1} T^{1-\alpha}
\end{align}
We have  $ p_1 + p_2 = 1 $, therefore since the logarithm is a concave function we can apply Jensen inequality:
\begin{equation}
 \sum_j p_j \log x_j \le \log \left( \sum_j p_j  x_j  \right) 
\end{equation}
to the second and third term in \eqref{Jensen}. Noticing that:
\begin{equation}
  \sum_j p_j  x_j = \frac{1}{(1-T)^\alpha + (d-1)^{1-\alpha} T^\alpha}
\end{equation}
this results in:
\begin{equation}
 \frac{d}{d \, \alpha} \left[ \frac{1}{1-\alpha} \log \left( (1-T)^\alpha + (d-1)^{1-\alpha} T^\alpha \right) \right] \le 0 
\end{equation}

\section{Lieb-Robinson bounds for strictly local Hamiltonians} \label{L-R_bounds}

The results of this section partially appeared in \cite{Osborne_slides,Osborne_2007}, we include it because the faster than exponential decrease of the spatial part of the Lieb-Robinson bounds in \eqref{superexp} is the crucial ingredient for the summation of the series \eqref{upper_bound_series}. $ H $ is as in \eqref{Ham}, $ A $ is any operator supported in $ [0,1] $. The idea that motivates the definition of $ \Delta_j(t) $ as below is to quantify how much the evolution in time of an operator is affected by terms in the Hamiltonian that are far away from the support of such operator. We also note that the RHS of \eqref{in} is as in equation (39) of \cite{Haah_2021}. The way in which a linear lightcone follows from equation \eqref{superexp} was given, for example, in \cite{Baldwin_2023}. The work \cite{Baldwin_2025} contains the proof that on-site terms in the Hamiltonian cannot increase (but can certainly decrease \cite{Sims_Stolz_2012,Elgart_2023}) the Lieb-Robinson velocity.

\begin{align}
& \Delta_j(t) := \| e^{iH_{\Lambda_{j+1}}t}Ae^{-iH_{\Lambda_{j+1}}t} - e^{iH_{\Lambda_j}t}Ae^{-iH_{\Lambda_j}t}  \| \label{in}  \\ 
 & = \|  \int_0^t ds \, \frac{d}{ds}  \left[ e^{iH_{\Lambda_{j+1}}(t-s)}  \left(  e^{iH_{\Lambda_j}s}Ae^{-iH_{\Lambda_j}s} \right) e^{-iH_{\Lambda_{j+1}}(t-s)}  \right] \| \\
 & = \|  \int_0^t ds \,  e^{iH_{\Lambda_{j+1}}(t-s)}   \Big[  H_{\Lambda_{j+1}} - H_{\Lambda_{j}} , e^{iH_{\Lambda_j}s}Ae^{-iH_{\Lambda_j}s}   \Big] e^{-iH_{\Lambda_{j+1}}(t-s)}   \| \label{comm}
\end{align}
We notice that $  H_{\Lambda_{j+1}} - H_{\Lambda_{j}} = H_{j,j+1} + H_{-j-1,-j} $ and $ H_{\Lambda_{j-1}} $ have disjoint supports, therefore they commute, hence we can insert $ e^{iH_{\Lambda_{j-1}}s}Ae^{-iH_{\Lambda_{j-1}}s} $ in the RHS of the commutator in \eqref{comm} provided $ [0,1] $ that is the support of $ A $ does not overlap with $ [j,j+1] $ and $ [-j-1,-j] $, the smallest $ j $ that ensures this to be true is $ j=2 $.
\begin{align}
 & \Delta_j(t) = \|  \int_0^t ds \,  e^{iH_{\Lambda_{j+1}}(t-s)}   \Big[  H_{j,j+1} + H_{-j-1,-j}, e^{iH_{\Lambda_j}s}Ae^{-iH_{\Lambda_j}s} - e^{iH_{\Lambda_{j-1}}s}Ae^{-iH_{\Lambda_{j-1}}s}  \Big] e^{-iH_{\Lambda_{j+1}}(t-s)} \| \label{B4} \\
 & \le 4 \max \{ \|  H_{j,j+1} \| \, , \| \, H_{-j-1,-j} \| \}  \int_0^t ds \, \| e^{iH_{\Lambda_j}s}Ae^{-iH_{\Lambda_j}s} - e^{iH_{\Lambda_{j-1}}s}Ae^{-iH_{\Lambda_{j-1}}s}  \| \label{rec} 
\end{align}
We recognize that the integrand in \eqref{rec} is $ \Delta_{j-1}(t) $, meaning that we have set up a recursive equation.
Defining $ J := \max_j \{ \| H_{j,j+1} \| \} $, \eqref{rec} is rewritten as 
\begin{align}
 \Delta_j(t) &\le 4 J \int_0^t ds \Delta_{j-1}(s) \\
 & \le (4 J)^{j-1}  \int_0^t dt_1 \int_0^{t_1} dt_2 ... \int_0^{t_{j-2}} dt_{j-1}   \Delta_1(t_{j-1})   
\end{align}
At this point we observe that, based on equations \eqref{in}-\eqref{comm}
\begin{align}
 \Delta_1(t) \le 4 J \|A\|  \, t
\end{align}
In conclusion:
\begin{align} \label{supp_1}
\Delta_j(t) & \le  \|A \| (4 J)^{j}  \int_0^t dt_1 \int_0^{t_1} dt_2 ... \int_0^{t_{j-2}} t_{j-1} dt_{j-1}  \le  \|A \| (4 J)^{j}  \frac{t^{j}}{j!}     
\end{align}
that renormalizing time as $ t':= 4 J t $, reads:
\begin{align} \label{superexp}
\Delta_j(t) \le  \|A \|  \frac{t'^{j}}{j!}     
\end{align}

Following \cite{Baldwin_2023} let us show how a linear lightcone emerges from \eqref{superexp}.
\begin{align} \label{to_linear}
  \frac{t'^{j}}{j!}   \le  \left( \frac{et'}{j} \right)^j =  \exp \left( j \log \frac{et'}{j} \right) \le   e^{et'-j}
\end{align}
The quantity $ v_{LR} := 4eJ $ might be called Lieb-Robinson velocity.

It turns out that an estimate on $ \Delta_j(t) $ is equivalent to the usual form of the Lieb-Robinson bounds. Let us sketch this equivalence.

Given $ A $ and $ B $ operators defined for simplicity on a unidimensional lattice system $ \Lambda_L:=[-L,L] $, with $ A $ supported on $ [0,1] $ and the support of $ B $ on $ [l+1,l+1+|b|] $, implying $\dist(\supp(A),\supp(B) ) = l $, we get:
\begin{align}
 \|[A(t),B]\| & = \|[A(t)- e^{iH_{\Lambda_l}t}Ae^{-iH_{\Lambda_l}t},B]\|  \le  2\|B\| \|A(t)- e^{iH_{\Lambda_l}t}Ae^{-iH_{\Lambda_l}t}\|
\end{align}
With $ A(t) = e^{iH_{\Lambda_L}t}Ae^{-iH_{\Lambda_L}t} $, we get:
\begin{align}
 A(t)- e^{iH_{\Lambda_l}t}Ae^{-iH_{\Lambda_l}t}  & = e^{iH_{\Lambda_L}t}Ae^{-iH_{\Lambda_L}t} - e^{iH_{\Lambda_{L-1}}t}Ae^{-iH_{\Lambda_{L-1}}t}  + ...  
 + e^{iH_{\Lambda_{l+1}}t}Ae^{-iH_{\Lambda_{l+1}}t} - e^{iH_{\Lambda_l}t}Ae^{-iH_{\Lambda_l}t} 
\end{align}
Then:
\begin{align}
  \|A(t)- e^{iH_{\Lambda_l}t}Ae^{-iH_{\Lambda_l}t}\| \le  \|A\| \sum_{j=l}^\infty  \frac{t'^j}{j!} 
\end{align}
\begin{align}
 & \sum_{j=l}^\infty  \frac{t'^j}{j!}  = t'^l  \sum_{j=l}^\infty  \frac{t'^{j-l}}{j!} = t'^l  \sum_{k=0}^\infty  \frac{t'^{k}}{(k+l)!}  = \frac{ t'^l}{l!}  \sum_{k=0}^\infty  \frac{t'^{k}}{{k+l \choose k}k!} \le \frac{ t'^l}{l!}  \sum_{k=0}^\infty  \left( \frac{t'}{l+1} \right)^k \label{last_sum} 
\end{align}
In the last step of \eqref{last_sum} we have used $ {k+l\choose k} \ge \frac{(l+1)^k}{k!} $. With $ t' < l+1 $ it follows:
\begin{align}
  \sum_{j=l}^\infty  \frac{t'^j}{j!} \le \frac{ t'^l}{l!} \frac{1}{1-\frac{t'}{l+1}}
\end{align}
With $ t' \le \frac{l+1}{2} $, it is $ \frac{1}{1-\frac{t'}{l+1}} \le 2 $, then overall, given that $ l:= \dist(\supp(A),\supp(B))$, it holds:
\begin{align}
 \|[A(t),B]\| \le 4 \|A\| \|B\| \frac{ t'^l}{l!} 
\end{align}
We stress that the condition  $ t' \le \frac{l+1}{2} $ is not restrictive, in the sense that the trivial bound to $ \|[A(t),B]\| $, for generic $ A $ and $ B $, is $ 2 \|A\| \|B\| $, then let us determined the time $ t' $ such that $ \frac{ t'^l}{l!} = \frac{1}{2} $, that with $ l $ large enough is
\begin{align}
  t'= \left( \frac{l!}{2} \right)^{\frac{1}{l}} \le (el)^{\frac{1}{l}} \frac{l}{e} \le \frac{l+1}{2}
\end{align}
To obtain a bound on $ \Delta_j(t) $ from a L-R bound is immediate considering equation \eqref{comm}.

We now consider, in one dimension, the more general case of an operator $ A $ with support contained in $ \Lambda_k := [-k,k] $. Moving the origin of the lattice this can always be the case. Let us consider $ j > k $. To upper bound $ \Delta_j(t) $ we repeat the same steps from   \eqref{in} to \eqref{rec}, noticing that the insertion in \eqref{B4} of the term $ e^{iH_{\Lambda_{j-1}}s}Ae^{-iH_{\Lambda_{j-1}}s} $ is allowed provided $ j > k $, meaning that the smallest value of $ j $ such that this is possible is $ j=k+1 $. This implies that the iteration continues till under the multiple time-integral we have $ \Delta_k $.
\begin{align}
 \Delta_j(t) &\le 4 J \int_0^t ds \Delta_{j-1}(s) \\
 & \le (4 J)^{j-k}  \int_0^t dt_1 \int_0^{t_1} dt_2 ... \int_0^{t_{j-k-1}} dt_{j-k}   \Delta_k(t_{j-k})   
\end{align}
At this point we observe that, based on equations \eqref{in}-\eqref{comm}
\begin{align}
 \Delta_k(t) \le 4 J \|A\|  \, t
\end{align}
In conclusion:
\begin{align} \label{supp_k}
 \Delta_j(t) 
& \le  \|A \| (4 J)^{j-k+1}  \int_0^t dt_1 \int_0^{t_1} dt_2 ... \int_0^{t_{j-k-1}} dt_{j-k} t_{j-k} \\
& \le  \|A \| (4 J)^{j-k+1}  \frac{t^{j-k+1}}{(j-k+1)!}     
\end{align}
The equation \eqref{supp_k} reduces to \eqref{supp_1} in the case $ k=1 $.
Let us consider an operator $ B $ supported at a distance $ l $ from $ A $, that we recall has support contained in $ [-k,k] $. Repeating the same approach as above we find that:
\begin{align} \label{general_L-R}
 \|[A(t),B]\| \le 4 \|A\| \|B\| \frac{ (4Jt)^l}{l!} 
\end{align}
The Lieb-Robinson bound \eqref{general_L-R} depends only on the norms of the operators $ A $ and $ B $, and $ J $, but not on the size of their supports.

The authors of the review \cite{Chen_2023} using graph theory \cite{Chen_2021} are able, see their equation 3.51, to obtain a better scaling with respect to time of the L-R bound that amounts to replace $ 4J $ with $ 2J $ in \eqref{general_L-R}.

For a more general approach that also applies to any dimension, see the derivation of Lieb-Robinson bound in appendix C of \cite{Masanes_2009}.

We conclude this section with two remarks. If the operator $ A $ is positive, $ A \ge 0 $, then the rescale time becomes $ t'=2\,J\,t$, in fact we can use $ \|A - \frac{\|A\|}{2}\mathds{1} \| = \frac{\|A\|}{2} $, from \cite{Hastings_2014}, for a better upper bound. With $ B $ a generic operator, we have:
\begin{align}
  \| [A,B] \| = \| [A - \frac{\|A\|}{2}\mathds{1},B] \| 
  \le 2\|A - \frac{\|A\|}{2}\mathds{1} \| \, \| B \| = \| A  \| \, \| B \|
\end{align}

A more general relation is actually true, namely, if $ A \ge 0 $, denoting $ \lambda_{max} $ and $ \lambda_{min} $ the largest and smallest eigenvalues of $ A $, then $ \|A - \lambda_{min}\mathds{1} - \frac{\lambda_{max}-\lambda_{min}}{2}\mathds{1} \| = \frac{\lambda_{max}-\lambda_{min}}{2} $. With $ \lambda_{min}=0 $ we recover the previous one.

The second remark regards the extension of \eqref{superexp} to any Schatten norm, that, with $ |A|:=\sqrt{A^*A} $, is defined as:
\begin{align}
 \|A\|_p := \left(\Tr |A|^p \right)^{\frac{1}{p}}
\end{align}
Using H\"{o}lder's inequality, we have $ \|A \, B\|_p \le \|A \|_p \| B\| $. This implies:
\begin{align}
 \| e^{iH_{\Lambda_{j+1}}t}Ae^{-iH_{\Lambda_{j+1}}t} - e^{iH_{\Lambda_j}t}Ae^{-iH_{\Lambda_j}t}  \|_p \le  \|A \|_p  \frac{t'^{j}}{j!}    
\end{align}
This is relevant when considering quantities like the so called out-of-time-correlators (OTOC): $ \Tr \left( \sigma  A(t)BA(t)B \right) $. Denoting $ \sigma=\frac{\mathds{1}}{d}$ the maximally mixed state, with $ A=A^* $, $ B=B^*$, $ A^2=B^2=\mathds{1} $, and with $ j $ the distance among the supports of $ A $ and $ B $ we get:
\begin{align}
  \frac{1}{d} \| [A(t),B] \|_2^2 := \frac{1}{d} \Tr \left( [A(t),B]^* [A(t),B] \right) = 2 - \frac{2}{d}  \Tr \left(  A(t)BA(t)B \right) \le \frac{1}{d} \left(4 \|A \|_2  \frac{t'^{j}}{j!} \right)^2 =  \left(  4 \frac{t'^{j}}{j!} \right)^2
\end{align}
\begin{align}
 \frac{1}{d} \Tr \left(  A(t)BA(t)B \right) \ge 1 - 8 \left(   \frac{t'^{j}}{j!} \right)^2 \ge 1 -  \left(  8\left( \frac{e t'}{j} \right)^j \right)^2   \label{lower_OTOC}
\end{align}
At $ t=0 $, $ A $ and $ B $ commutes, moreover, under the assumption $ A^2=B^2=\mathds{1} $, it holds: $ \frac{1}{d} \Tr \left(  ABAB \right) = \frac{1}{d} \Tr \left(  A^2B^2 \right) = 1 $, then the bound \eqref{lower_OTOC} is saturated. Also:
\begin{align} \label{OTOC_upper}
 & \frac{1}{d} | \Tr \left(  A(t)BA(t)B \right) | \le \frac{1}{d}  \Tr |  A(t)BA(t)B | \\ 
 & \le \frac{1}{d} \| B \| \|A(t)BA(t) \|_1  \le \frac{1}{d} \| B \| \|A^2(t)B \|_1 \le \frac{1}{d} \| B \|^2 \|A^2(t) \|_1 = 1 
\end{align}
In \eqref{OTOC_upper} we have used the fact that for two matrices $ C $ and $ D $ such that the product $ CD $ is Hermitean then: $  \|CD \|_1  \le \|DC\|_1 $, that is theorem 8.1 of \cite{Simon_Trace}.

The interpretation of \eqref{lower_OTOC} is straightforward: the decrease of $ \Tr \left(  A(t)BA(t)B \right) $ signals the increase of the overlap among the supports of $ A(t) $ and $ B $ starting from a time scale of the order of the distance among the supports of $ A $ and $ B $. 
We stress that with $ A \ge 0 $ and $ B \ge 0 $ both positive then $  \Tr \left(  A(t)BA(t)B \right) $ is positive as well, is fact a product of positive operators is positive.

\section{Evaluating the trace distance \eqref{del_rho}} \label{bound_trace_dist}

We now upper bound
\begin{align} \label{delta_rho}
 & T_{k}(t) := \frac{1}{2} \|\Tr_{[1,k]}  V_{\Lambda_{k}}(t) \rho_{[-k,k]}  V^*_{\Lambda_{k}}(t) - \Tr_{[1,k]} V_{\Lambda_{k-1}}(t) \rho_{[-k,k]}  V^*_{\Lambda_k-1}(t) \|_1 \\
 & = \frac{1}{2} \|\Tr_{[1,k]}  V_{\Lambda_{k}}(t) \rho  V^*_{\Lambda_{k}}(t) - \Tr_{[1,k]} V_{\Lambda_{k-1}}(t) \rho  V^*_{\Lambda_k-1}(t) \|_1
\end{align}
To do so we employ the following upper bound where $ U(t) $ is a generic time-dependent unitary operator with 
$ U(0)= \mathds{1}_{\mathcal{H}} $ and $ \sigma $ a generic density matrix.
\begin{align}
 & \|  U(t) \sigma U^*(t) - \sigma \|_1 = \| \int_0^t ds \, \frac{d}{ds} \left( U(s)\sigma U^*(s) \right) \|_1 \\
 & =  \| \int_0^t ds \, \left[ \left( \frac{d}{ds}  U(s) \right) \sigma U^*(s) + U(s) \sigma \left( \frac{d}{ds}  U^*(s) \right) \right] \|_1 \\
 &= \| \int_0^t ds \, \Big[ \left( \frac{d}{ds}  U(s) \right) \sigma U^*(s) - U(s) \sigma U^*(s) \left( \frac{d}{ds}  U(s) \right) U^*(s) \Big] \|_1 \\
 &= \| \int_0^t ds \, U(s) \Big[ U^*(s) \left( \frac{d}{ds}  U(s) \right) \sigma  - \sigma U^*(s) \left( \frac{d}{ds}  U(s) \right) \Big] U^*(s)  \|_1 \\
 & \le  \int_0^{|t|} ds \| \left[ U^*(s)  \frac{d}{ds}  U(s)  , \sigma  \right] \|_1 \\
 & \le 2 \int_0^{|t|} ds \|  U^*(s)  \frac{d}{ds}  U(s)   \|  \\
 & \le 2|t| \sup_{s\in[0,|t|]}  \| U^*(s)  \frac{d}{ds} U(s) \|  \label{up_com}
\end{align}

We now apply \eqref{up_com} to \eqref{delta_rho}.

\begin{align}
  & \|  V_{\Lambda_{k+1}}(t) \rho  V^*_{\Lambda_{k+1}}(t) - V_{\Lambda_{k}}(t) \rho  V^*_{\Lambda_k}(t) \|_1 \nonumber \\
  & =  \|  V_{\Lambda_{k}}^*(t)  V_{\Lambda_{k+1}}(t) \rho  V^*_{\Lambda_{k+1}}(t) V_{\Lambda_{k}}(t)  - \rho  \|_1 \\
  & \le 2 \int_{s=0}^{|t|} ds \| V^*_{\Lambda_{k+1}}(s) V_{\Lambda_{k}}(s) \frac{d}{ds} \left(  V_{\Lambda_{k}}^*(s)  V_{\Lambda_{k+1}}(s) \right) \| \\
  & = 2 \int_{s=0}^{|t|} ds \| V^*_{\Lambda_{k+1}}(s) V_{\Lambda_{k}}(s) \left( \frac{d}{ds}   V_{\Lambda_{k}}^*(s) \right) V_{\Lambda_{k+1}}(s) + V^*_{\Lambda_{k+1}}(s) \frac{d}{ds}  V_{\Lambda_{k+1}}(s) \| \label{der_V}
\end{align}
Recalling that $  V_{\Lambda_k}(s):= e^{is (H_{\Lambda_k} -H_I)}  e^{-is H_{\Lambda_k}} $, where we have defined $ H_I := H_{[0,1]} $, we have that:
\begin{align} \label{int_picture_NN}
 i \frac{d}{ds}  V_{\Lambda_{k}}(s) = e^{is (H_{\Lambda_k} -H_I)} H_I e^{-is (H_{\Lambda_k} -H_I)}V_{\Lambda_{k}}(s)
\end{align}
Then:
\begin{align}
  & T_{k+1}(t) \le  \int_{s=0}^{|t|} ds \| e^{is (H_{\Lambda_k} -H_I)} H_I e^{-is (H_{\Lambda_k} -H_I)} -  e^{is (H_{\Lambda_{k+1}} -H_I)} H_I e^{-is (H_{\Lambda_{k+1}} -H_I)} \|  =  \int_{s=0}^{|t|} ds \Delta_k(s) \label{C12} \\
  & \le   \int_{s=0}^{|t|} ds \|H_I\|  \frac{(4 J s)^{k}}{k!} \le  \frac{1}{4} \frac{t'^{k+1}}{(k+1)!} \label{final_trace_dist}
\end{align}
We remark that the dynamics of $ \Delta_k(t) $ is here generated by $ H_{\Lambda_k} - H_I $, but all the steps that lead to \eqref{superexp} are equally valid. To obtain \eqref{final_trace_dist} we have used \eqref{superexp}, the definition of $ J := \sup \{ \|H_{j,j+1}\| \} $, that implies $ \| H_I \| \le J $, and finally the definition of $ t':= 4Jt $.

Inserting \eqref{C12} into \eqref{bound_sum_1} we obtain \eqref{bound_sum_intro}.


\section{Brief discussion of $ u_k(t) $} \label{on_u}

We have defined $ u_k(t) := \frac{1}{4}\left(\frac{e \, t}{k} \right)^k \ge \frac{1}{4} \frac{t^k}{k!} $. $ u_k(t) $ upper bounds the trace distance. With a fixed $ k $, a sufficient condition for $ u_k(t) \le 1 $ is given by $ t \le \frac{k}{e} $. 

It is easy to see that, with a fixed $ t $, $ u_k(t) $ is a decreasing function of $ k $, in fact:
\begin{align} \label{u_decr}
 \frac{u_{k+1}(t)}{u_{k}(t)} = \frac{\left(\frac{e \, t}{k+1}\right)^{k+1}}{\left(\frac{e \, t}{k}\right)^{k}}=\frac{e \, t}{k+1} \left( \frac{k}{k+1}  \right)^k < \frac{e \, t}{k+1} 
\end{align}
The condition $ t \le \frac{k}{e} $, that implies $ u_k(t) < 1 $, also ensures that $ u_k(t) $ is a decreasing function of $ k $.

\section{Algorithm for the minimization of the sum that upper bounds $ \Delta S_\alpha(t) $} \label{alg}

In this section we discuss the existence and uniqueness of the solution $ l_\alpha(t) $ that satisfies the two equations \eqref{system_min}. We consider $ \alpha \in (0,1) $, the value $ \alpha=1 $ can be dealt with analogously. 
\begin{itemize}

\item
We fix $ t $  and consider a value $ \bar{l} $ of $ l $ such that both the LHSs of the equations in \eqref{system_min} are larger equal than $ \log r $, this is certainly the case with $ u_{\bar{l}} (t) \in [0,1] $ close enough to $ 1 $.

\item We increase $ l $ from $ \bar{l} $ to $ \bar{l} +1 $. We remark the crucial fact that the value taken from the LHS of the second equation is now the value that the LHS of the first equation in \eqref{system_min} had at the previous step. Therefore they cannot jump in one step from being both larger than $ \log r $ to be both smaller than $ \log r $. Since $ u_l(t) \le 1 $ is a decreasing function of $ l $ and both the LHSs of \eqref{system_min} are increasing in $ u $ with $ u_l(t) \in [0,1] $, they will both decrease going from $ \bar{l} $ to $ \bar{l} +1 $.

\item We will repeat increasing $ l $ till the LHS of the first equation will become smaller than $ \log r $. This will certainly happen at a certain point, in fact with $ u $ of the order of $ r^{-(l+2)\frac{1-\alpha}{\alpha}} $  the LHS of the first equation in \eqref{system_min} will be of order $ \log r $. This is the value of $ l $ that minimizes \eqref{bound_sum_1}. 
\end{itemize}

\section{Estimate of the series in \eqref{bound_sum_3}} \label{sum_details}

We want to estimate the series in \eqref{bound_sum_3}, that is copied below \eqref{bound_sum_3_copy},  using  $ l_\alpha(t') := c \, r^{\frac{1-\alpha}{\alpha}} \, t'  $, with $ c > e $ as in \eqref{l_alpha}, and $ u_k(t'):= \frac{1}{4} \left( \frac{e t'}{k} \right)^k $. The result will be that for every $ t'>0 $ and $ \alpha \in (0,1] $ \eqref{bound_sum_3_copy} is upper bounded by a constant of order $ 1 $, and is exponentially decreasing in $t'$.

\begin{equation}\label{bound_sum_3_copy}
\frac{1}{1-\alpha}  \sum_{k={l_\alpha+1}}^{L} \log \left[ 1 - \alpha u_k(t') + r^{(k+1)(1-\alpha)} u_k(t')^\alpha \right]
\end{equation}

We have already stated and used the fact that the Audenaert upper bound on $ \Delta S_\alpha $ is increasing in the trace distance. Let us discuss more in details that the argument of the $ \log $ in \eqref{bound_sum_3_copy} is an increasing function of $ 0 \le u \le 1 $, at a fixed $k$.
\begin{align}
0 = \frac{\partial}{\partial u} \left[ 1- \alpha u + r^{(k+1)(1-\alpha)} u^\alpha \right] = -\alpha + \alpha r^{(k+1)(1-\alpha)} u^{\alpha-1} \label{extremum}
\end{align}
That has the solution, corresponding to a maximum, as it is seen by the always negative second derivative, $ u = r^{k+1} $. $ r^{k+1} $ has minimal value, reached for $k=1$, (being $ k=l_\alpha(t) + 1 $, $k=1$ is given at $ t=0$), equal to $ r^2 $. The minimal value of $ r $ is $ 2 $, therefore, in the interval $ 0 \le u \le 1 $, the RHS of \eqref{bound_sum_3_copy} is an increasing function of $ u $. It means that replacing $ u $ with an upper bound (smaller than $ r^2 $) we still obtain an upper bound, despite the negative term $ - \alpha u $ in the argument of the $ \log $ in \eqref{bound_sum_3_copy}.

Let us change variable in the summation in \eqref{bound_sum_3_copy} setting $ k = l_\alpha + n $, then:
\begin{align}
 u_{l_\alpha + n}=\frac{1}{4}\left(\frac{e}{cr^\frac{1-\alpha}{\alpha}}\right)^{l_\alpha + n} \left(\frac{l_\alpha}{l_\alpha+n}\right)^{l_\alpha+n}
\end{align} 
The function $ \left(\frac{x}{x+n}\right)^{x+n} = \left(1-\frac{n}{x+n}\right)^{x+n}  $ is increasing in $ x $, with $ \lim_{x\rightarrow \infty} \left(1-\frac{n}{x+n}\right)^{x+n} = e^{-n}$. With $ x=l_\alpha $, the limit $ t \rightarrow \infty $ corresponds to $ l_\alpha \rightarrow \infty $. Then:
\begin{align}
 u_{l_\alpha + n} \le \frac{1}{4}\left(\frac{e}{cr^\frac{1-\alpha}{\alpha}}\right)^{l_\alpha + n}e^{-n} = \frac{1}{4}\left(\frac{e}{cr^\frac{1-\alpha}{\alpha}}\right)^{l_\alpha}\left(\frac{1}{cr^\frac{1-\alpha}{\alpha}}\right)^{n}
\end{align}
It follows that:
\begin{align}
  & r^{(l_\alpha + n+1)(1-\alpha)} u_{l_\alpha + n}^\alpha \le 
  \frac{r^{1-\alpha}}{4^\alpha} \left( \frac{e}{c} \right)^{ \alpha l_\alpha} \left( \frac{1}{c^\alpha} \right)^n 
\end{align}
We then upper bound \eqref{bound_sum_3_copy}:
\begin{align}
&\frac{1}{1-\alpha} \sum_{n=1}^{\infty} \log \left[ 1- \alpha u_{l_\alpha+n}(t') + r^{(l_\alpha+n+1)(1-\alpha)} u_{l_\alpha+n}(t')^\alpha \right] \label{var_change} \\
& \le \frac{1}{1-\alpha} \sum_{n=1 }^{\infty}  \Big[ - \alpha \frac{1}{4} \left(\frac{e}{c r^\frac{1-\alpha}{\alpha}}\right)^{l_\alpha} \left(\frac{1}{cr^\frac{1-\alpha}{\alpha}}\right)^{n}  + \frac{r^{(1-\alpha)}}{4^\alpha} \left(\frac{e}{c}\right)^{\alpha l_\alpha } \left( \frac{1}{c^\alpha} \right)^n  \Big] \\
 & = \frac{1}{1-\alpha}   \Big[ - \alpha \frac{1}{4} \left(\frac{e}{c r^\frac{1-\alpha}{\alpha}}\right)^{cr^\frac{1-\alpha}{\alpha}t' } \frac{1}{cr^\frac{1-\alpha}{\alpha}-1} + \frac{r^{(1-\alpha)}}{4^\alpha} \left(\frac{e}{c}\right)^{\alpha cr^\frac{1-\alpha}{\alpha}t' } \frac{1}{c^\alpha-1}\Big] \label{final_upper_sum}
\end{align}
With $ c > e $, it possible to see that for all $ t > 0 $, in the limit $ \alpha \rightarrow 0 $ the upper bound \eqref{final_upper_sum} tends to zero.

The factor in between square brackets in \eqref{final_upper_sum} is a differentiable function of $ \alpha $ and is vanishing with $ \alpha = 1$. This means that the limit $ \alpha \rightarrow 1 $ of \eqref{final_upper_sum} equals minus the first derivative in $ \alpha = 1$
of such function. An explicit calculation reveals that the limit $ \alpha \rightarrow 1 $ of \eqref{final_upper_sum} is again exponentially decreasing in $ t $.

Denoting $ K'-\log r $ the upper bound in \eqref{final_upper_sum}, from \eqref{bound_sum_3}  it follows:
\begin{equation} 
 \Delta S_{\alpha}(t) \le l_\alpha(t') \, \log r + K'
\end{equation}
Then:
\begin{equation} 
 \Delta S_{\alpha}(t) \le  c \, r^{\frac{1}{\alpha}-1} \, t' \, \log r  + K'
\end{equation}
Recollecting that $ t':= 4 \, t J $, we obtain, with $ K := 4  \, c $
\begin{equation} 
 \Delta S_{\alpha}(t) \le K \, r^{\frac{1}{\alpha}-1}  \, J \, t \, \log r + K'
\end{equation}
as in \eqref{entropy_bound}.


\section{Accounting for tails of interactions in the Hamiltonian for an upper bound of $\alpha$-Rényi entropy} \label{sub_bound_exp_tail}

In this section we provide the details to obtain the upper bound \eqref{trace_dist_tails}.
Applying \eqref{der_V} we obtain:
\begin{align}
  & \|  {\widehat V}_{\Lambda_{k+1}}(t) \rho {\widehat V}^*_{\Lambda_{k+1}}(t) -  {\widehat V}_{\Lambda_{k}}(t) \rho {\widehat V}^*_{\Lambda_{k}}(t) \|_1 \le 2 t \sup_{s\in[0,t]} \| {\widehat V}^*_{\Lambda_{k+1}}(s) {\widehat V}_{\Lambda_{k}}(s) \frac{d}{ds} \left(  {\widehat V}_{\Lambda_{k}}^*(s) \right) {\widehat V}_{\Lambda_{k+1}}(s) + {\widehat V}^*_{\Lambda_{k+1}}(s) \frac{d}{ds}  {\widehat V}_{\Lambda_{k+1}}(s) \| \label{der_V_hat}
\end{align}
Recalling that 
\begin{align}
  \widehat{V}_{\Lambda_k}(t) := e^{it \left( \widehat{H}_{[-k,-b]} + \widehat{H}_{[b,k]} \right)}   e^{-it \widehat{H}_{[-k,k]}}
\end{align}
it holds:
\begin{align} \label{int_picture_tails}
 i \frac{d}{ds}  \widehat{V}_{\Lambda_{k}}(s) = e^{it \left( \widehat{H}_{[-k,-b]} + \widehat{H}_{[b,k]} \right)} \left( \widehat{H}_{[-k,k]} - \widehat{H}_{[-k,-b]} - \widehat{H}_{[b,k]} \right) e^{-it \left( \widehat{H}_{[-k,-b]} + \widehat{H}_{[b,k]} \right)} \widehat{V}_{\Lambda_{k}}(s)
\end{align}
Comparing this equation with \eqref{int_picture_NN}, where the generator of $ V_{\Lambda_k} $ was a unitary conjugation of the same Hermitean operator, $ H_I $, for all $ k $, we see that in \eqref{int_picture_tails} instead we have the operator $  \widehat{H}_{[-k,k]} - \widehat{H}_{[-k,-b]} - \widehat{H}_{[b,k]} $ that, despite decaying outside of the interval $ [-b,b] $, is supported on $ [-2k,2k] $, therefore it is $k$-dependent. All together we have:  
\begin{align}
  & \|  {\widehat V}_{\Lambda_{k+1}}(t) \rho {\widehat V}^*_{\Lambda_{k+1}}(t) -  {\widehat V}_{\Lambda_{k}}(t) \rho {\widehat V}^*_{\Lambda_{k}}(t) \|_1 \label{97} \\ 
  &\le 2 t \sup_{s\in[0,t]} \|  e^{is \left( \widehat{H}_{[-k,-b]} + \widehat{H}_{[b,k]} \right)} \left( \widehat{H}_{[-k,k]} - \widehat{H}_{[-k,-b]} - \widehat{H}_{[b,k]} \right) e^{-is \left( \widehat{H}_{[-k,-b]} + \widehat{H}_{[b,k]} \right)} +  \nonumber \\
  & - e^{is \left( \widehat{H}_{[-k-1,-b]} + \widehat{H}_{[b,k+1]} \right)} \left( \widehat{H}_{[-k-1,k+1]} - \widehat{H}_{[-k-1,-b]} - \widehat{H}_{[b,k+1]} \right) e^{-is \left( \widehat{H}_{[-k-1,-b]} + \widehat{H}_{[b,k+1]} \right)}  \|  \\
  & \le 2 t \sup_{s\in[0,t]} \|  e^{is \left( \widehat{H}_{[-k,-b]} + \widehat{H}_{[b,k]} \right)} \left( \widehat{H}_{[-k,k]} - \widehat{H}_{[-k,-b]} - \widehat{H}_{[b,k]} \right) e^{-is \left( \widehat{H}_{[-k,-b]} + \widehat{H}_{[b,k]} \right)} +  \nonumber \\
  & - e^{is \left( \widehat{H}_{[-k-1,-b]} + \widehat{H}_{[b,k+1]} \right)} \left( \widehat{H}_{[-k,k]} - \widehat{H}_{[-k,-b]} - \widehat{H}_{[b,k]} \right) e^{-is \left( \widehat{H}_{[-k-1,-b]} + \widehat{H}_{[b,k+1]} \right)}  \| +  \nonumber \\
  & + 2 t \, \|  \widehat{H}_{[-k-1,k+1]} - \widehat{H}_{[-k-1,-b]} - \widehat{H}_{[b,k+1]}  -  \left( \widehat{H}_{[-k,k]} - \widehat{H}_{[-k,-b]} - \widehat{H}_{[b,k]} \right) \| \label{99}
\end{align}
We see that the first operator norm of \eqref{99} has almost the structure of a $ \Delta_k(s) $ as in appendix \ref{L-R_bounds}, the difference with \eqref{in} is that, in the case of a strictly local Hamiltonian the operator whose Heisenberg evolutions are computed had a fixed support, while in \eqref{99} it depends, despite only with tails, on $ k $.
We recollect that with $ r $ inside the region $ X $, that for simplicity is assumed connected, it holds:
\begin{equation} \label{G7}
\big \| \frac{1}{2^{|X^c|}} \left( \Tr_{X^c} H_r \right) \otimes \mathds{1}_{X^c} - H_r \big \| \le  J \,  f \left( \frac{\textrm{dist}(r,X^c)}{\xi} \right)
\end{equation}

Let us show that the second operator norm of \eqref{99} is bounded by $ O\left(J C_\xi f\left(\frac{k}{\xi}\right)\right) $, with $ f $ denoting the rate of decrease of interactions, as defined in \eqref{decay}. We start considering $ \widehat{H}_{[b,k+1]} - \widehat{H}_{[b,k]} $. To shorten the equations we define: $ \widetilde{\Tr}_{[m,n]^c}(A) := \frac{1}{2^{|[m,n]^c|}}  \Tr_{[m,n]^c} (A) \otimes \mathds{1}_{[m,n]^c} $. Let us assume for now $ 2k < L $, the case $ 2k \ge L $ is discussed in \eqref{104}. 
\begin{align}
 &\widehat{H}_{[b,k+1]} - \widehat{H}_{[b,k]} :=   \widetilde{\Tr}_{([1,2(k+1)]\cap[-L,L])^c}  \sum_{r=b}^{k+1} H_r  -   \widetilde{\Tr}_{([1,2k]\cap[-L,L])^c}  \sum_{r=b}^k H_r  \\
 & =  \widetilde{\Tr}_{([1,2(k+1)]\cap[-L,L])^c} \left( \mathds{1} -  \widetilde{\Tr}_{ ([1,2k]\cap[-L,L])^c \setminus ([1,2(k+1)]\cap[-L,L])^c } \right) \sum_{r=b}^{k} H_r   +  \widetilde{\Tr}_{([1,2(k+1)]\cap[-L,L])^c}   H_{k+1} \label{101}
\end{align}
The norm of the first term in \eqref{101} is small, of order $  J f\left( \frac{k}{\xi} \right) $, see \eqref{103}, in fact being the complements of the sets evaluated with respect to $ [-L,L] $:
\begin{align}
 ([1,2k]\cap[-L,L])^c \setminus ([1,2(k+1)]\cap[-L,L])^c = [2k+1,2(k+1)]\cap[-L,L] 
\end{align}
 the Hamiltonian term of $ \sum_{r=b}^{k} H_r $ with the closest centre to $ [2k+1,2(k+1)]\cap[-L,L] $ is $ H_k $. We also observe that the normalized trace defined above is a projection, see \cite{Lu_2024}, namely: $ \widetilde{\Tr}_{[m,n]^c} \left( \widetilde{\Tr}_{[m,n]^c}(A) \right) = \widetilde{\Tr}_{[m,n]^c}(A) $, therefore, as a superoperator, it has norm equal to $ 1 $. This implies that:
\begin{align}
& \|  \widetilde{\Tr}_{([1,2(k+1)]\cap[-L,L])^c} \left( \mathds{1} -  \widetilde{\Tr}_{([2k+1,2(k+1)]\cap[-L,L])} \right) \sum_{r=b}^{k} H_r \| \le \| \left( \mathds{1} -  \widetilde{\Tr}_{ ([2k+1,2(k+1)]\cap[-L,L]) } \right) \sum_{r=b}^{k} H_r  \| \label{102}  \\
& \le   J \sum_{r=b}^{k} f\left(\frac{2k+1-r}{\xi}\right)  \le   J C_\xi \left( f\left(\frac{k}{\xi}\right) - f\left(\frac{2k-b}{\xi}\right) \right) \label{103}
\end{align}
with $ C_\xi $ a function of $ \xi $. If, for example, the interactions have exponential tails, it is $ C_\xi = \xi $, in fact:
\begin{align}
 J \sum_{r=b}^{k} f\left(\frac{2k+1-r}{\xi}\right) & = J \sum_{r=b}^{k} e^{-\frac{2k+1-r}{\xi}} = J \sum_{l=k+1}^{2k+1-b} e^{-\frac{l}{\xi}} = J \sum_{l=0}^{2k+1-b} e^{-\frac{l}{\xi}} - J \sum_{l=0}^{k} e^{-\frac{l}{\xi}} \\
& = J \left( \frac{1-e^{-\frac{2k+2-b}{\xi}}}{1-e^{-\frac{1}{\xi}}} -  \frac{1-e^{-\frac{k+1}{\xi}}}{1-e^{-\frac{1}{\xi}}} \right) 
 \le J \xi \left( e^{-\frac{k}{\xi}} - e^{-\frac{2k-b}{\xi}} \right)
\end{align}
In the last line we have used $ e^{\frac{1}{\xi}}-1 \ge  \frac{1}{\xi} $

It is important to discuss what happens with $ 2k \ge L $, in this case we must recollect that the support, for example, of $ 
\widehat{H}_{[b,k]} $ is $ [1,2k] \cap [-L,L] $, therefore in \eqref{101} we have: 
\begin{align}
 & \left( [1,2k] \cap [-L,L] \right)^c \setminus \left( [1,2(k+1)]  \cap [-L,L] \right)^c  = [-L,0] \setminus [-L,0] =  \emptyset \label{104}
\end{align}
The partial trace associated with the empty set is the identity, therefore the term among parenthesis in \eqref{101} is vanishing.

We mention the fact that within this formalism we are able to recover the results for nearest neighbor Hamiltonians that we obtained in the first part of this paper, and actually to generalize those to the case of a $k$-neighbor Hamiltonian. In fact with $ f $ compactly supported and $ \xi $ the size of its support, the RHS of \eqref{G7} is vanishing. More in detail, considering
\begin{align} \label{f_compact}
 f\left(\frac{j}{\xi}\right) = \delta \left( \lfloor \frac{ j - \frac{1}{2}}{\xi} \rfloor , 0 \right)
\end{align}
$ \lfloor a \rfloor $ denotes the largest integer such that $ \lfloor a \rfloor \le a $. $ \delta(\cdot,\cdot) $ is the Kronecker delta. We take $ \xi $ integer, and $ \xi \ge 1 $. From the assumption $ r \in X $ in \eqref{G7}, follows that $ \dist(r,X^c) \ge 1 $, then $ j \ge 1 $ and integer. According to \eqref{f_compact}, with  $ j \le \xi $, it is $ f\left(\frac{j}{\xi}\right) = 1 $, with $ j > \xi $, $ f\left(\frac{j}{\xi}\right) = 0 $. This also shows that the value of $ b $, that we have introduced in \eqref{first}, for the case of a strictly local Hamiltonian that is the sum of $ \xi$-neighbors terms  equals $ \xi $.

We are now ready to upper bound the operator norm of the second term in equation \eqref{99}. To ease the notation all the sets are meant to be the intersection with $ [-L,L] $.
\begin{align}
 & \|  \widehat{H}_{[-k-1,k+1]} - \widehat{H}_{[-k-1,-b]} - \widehat{H}_{[b,k+1]}  -  \left( \widehat{H}_{[-k,k]} - \widehat{H}_{[-k,-b]} - \widehat{H}_{[b,k]} \right) \| \label{G16}  \\
 & = \| \left( \widetilde{\Tr}_{[-2(k+1),2(k+1)]^c} \left( \mathds{1} -  \widetilde{\Tr}_{ [-2k,2k]^c \setminus [-2(k+1),2(k+1)]^c } \right) \sum_{r=-k}^{k} H_r \right)  +  \widetilde{\Tr}_{[-2(k+1),2(k+1)]^c} \left( H_{-(k+1)} + H_{k+1} \right) \nonumber \\
 & -\left( \widetilde{\Tr}_{[-2(k+1),-1]^c} \left( \mathds{1} -  \widetilde{\Tr}_{ [-2k,-1]^c \setminus [-2(k+1),-1]^c } \right) \sum_{r=-k}^{-b} H_r \right)  -  \widetilde{\Tr}_{[-2(k+1),-1]^c}   H_{-(k+1) } \nonumber \\
 & -\left( \widetilde{\Tr}_{[1,2(k+1)]^c} \left( \mathds{1} -  \widetilde{\Tr}_{ [1,2k]^c \setminus [1,2(k+1)]^c } \right) \sum_{r=b}^{k} H_r \right)  -  \widetilde{\Tr}_{[1,2(k+1)]^c}   H_{k+1} \| \label{108}
\end{align}
Combining together in \eqref{108} the partial traces on $ H_{k+1} $ and $ H_{-(k+1)} $, using the same procedure as for \eqref{101}, and employing \eqref{103}, we obtain that \eqref{108} is upper bounded by $ O\left( J C_\xi f\left(\frac{k}{\xi}\right)\right) $.

More precisely from \eqref{108} we have five terms to keep into account. Term I:
\begin{align}
 & \| \left( \widetilde{\Tr}_{[-2(k+1),2(k+1)]^c} \left( \mathds{1} -  \widetilde{\Tr}_{ [-2k,2k]^c \setminus [-2(k+1),2(k+1)]^c } \right) \sum_{r=-k}^{k} H_r \right) \| \le \| \left( \mathds{1} -  \widetilde{\Tr}_{ [-2k-2,-2k-1] \cup [2k+1,2k+2] } \right) \sum_{r=-k}^{k} H_r  \| \nonumber \\
& \le \| \left( \mathds{1} -  \widetilde{\Tr}_{ [-2k-2,-2k-1] \cup [2k+1,2k+2] } \right) \sum_{r=0}^{k} H_r  \| + \| \left( \mathds{1} -  \widetilde{\Tr}_{ [-2k-2,-2k-1] \cup [2k+1,2k+2] } \right) \sum_{r=-k}^{0} H_r  \| \\
& \le 2 J \sum_{r=0}^k f\left(\frac{2k+1-r}{\xi} \right) \le 2 J C_\xi \left( f\left(\frac{k}{\xi}\right) - f\left(\frac{2k}{\xi}\right) \right)
\end{align}

Term II:
\begin{align}
 & \| \left( \widetilde{\Tr}_{[-2(k+1),-1)]^c} \left( \mathds{1} -  \widetilde{\Tr}_{ [-2k,-1]^c \setminus [-2(k+1),-1]^c } \right) \sum_{r=-k}^{-b} H_r \right) \| \le \| \left( \mathds{1} -  \widetilde{\Tr}_{ [-2k-2,-2k-1]  } \right) \sum_{r=-k}^{-b} H_r  \| \nonumber \\
& \le  J \sum_{r=b}^{k} f\left(\frac{2k+1-r}{\xi} \right) \le  J C_\xi \left( f\left(\frac{k}{\xi}\right) - f\left(\frac{2k-b}{\xi}\right) \right)
\end{align}

Term III:
\begin{align}
 & \| \left( \widetilde{\Tr}_{[1,2(k+1))]^c} \left( \mathds{1} -  \widetilde{\Tr}_{ [1,2k]^c \setminus [1,2(k+1)]^c } \right) \sum_{r=b}^{k} H_r \right) \| \le   J C_\xi \left( f\left(\frac{k}{\xi}\right) - f\left(\frac{2k-b}{\xi}\right) \right)
\end{align}

Term IV:
\begin{align}
 & \|  \widetilde{\Tr}_{[-2(k+1),2(k+1)]^c} H_{k+1} -  \widetilde{\Tr}_{ [1,2(k+1)]^c } H_{k+1} \| = \|  \widetilde{\Tr}_{[-2(k+1),2(k+1)]^c} \left( \mathds{1} -  \widetilde{\Tr}_{ [1,2(k+1)]^c \setminus [-2(k+1),2(k+1)]^c } \right)  H_{k+1} \| \nonumber \\
 & \le \| \left( \mathds{1} -  \widetilde{\Tr}_{ [-2(k+1),0]} \right) H_{k+1}  \| \le J f\left(\frac{k+1}{\xi}\right)  
\end{align}

Term V:
\begin{align}
 & \|  \widetilde{\Tr}_{[-2(k+1),2(k+1)]^c} H_{-(k+1)} -  \widetilde{\Tr}_{ [-2(k+1),-1]^c } H_{-(k+1)} \|  \le J f\left(\frac{k+1}{\xi}\right)  
\end{align}
Overall: 
\begin{align}
\textrm{Term I} + ... + \textrm{Term V}  \le 4 J C_\xi \left( f\left(\frac{k}{\xi}\right) - f\left(\frac{2k}{\xi}\right) \right) + 2 J f\left(\frac{k+1}{\xi}\right)  \label{second_term}
\end{align}

So far with \eqref{second_term} we have upper bounded the second term in equation \eqref{99}, as mentioned above the first term is in relation with the object that in appendix \ref{L-R_bounds} is called $ \Delta_k(t) $. In fact, from the first term in \eqref{99}:
\begin{align}
 & \| e^{is \left( \widehat{H}_{[-k,-b]} + \widehat{H}_{[b,k]} \right)} \left( \widehat{H}_{[-k,k]} - \widehat{H}_{[-k,-b]} - \widehat{H}_{[b,k]} \right) e^{-is \left( \widehat{H}_{[-k,-b]} + \widehat{H}_{[b,k]} \right)} +  \nonumber \\
  & - e^{is \left( \widehat{H}_{[-k-1,-b]} + \widehat{H}_{[b,k+1]} \right)} \left( \widehat{H}_{[-k,k]} - \widehat{H}_{[-k,-b]} - \widehat{H}_{[b,k]} \right) e^{-is \left( \widehat{H}_{[-k-1,-b]} + \widehat{H}_{[b,k+1]} \right)} \|   \\
  =\| \int_0^s du & \frac{d}{du} \Big[  e^{i(s-u) \left( \widehat{H}_{[-k-1,-b]} + \widehat{H}_{[b,k+1]} \right)} 
 e^{iu \left( \widehat{H}_{[-k,-b]} + \widehat{H}_{[b,k]} \right)} 
 \left( \widehat{H}_{[-k,k]} - \widehat{H}_{[-k,-b]} - \widehat{H}_{[b,k]} \right) \nonumber  \\
 & e^{-iu \left( \widehat{H}_{[-k,-b]} + \widehat{H}_{[b,k]} \right) } 
  e^{-i(s-u) \left( \widehat{H}_{[-k-1,-b]} + \widehat{H}_{[b,k+1]} \right)} \Big] \| \\
  \le \int_0^s du & \| \Big[ \widehat{H}_{[-k-1,-b]} + \widehat{H}_{[b,k+1]} - 
  \widehat{H}_{[-k,-b]} - \widehat{H}_{[b,k]} ,  \nonumber \\ 
  & e^{iu \left( \widehat{H}_{[-k,-b]} + \widehat{H}_{[b,k]} \right)}
  \left( \widehat{H}_{[-k,k]} - \widehat{H}_{[-k,-b]} - \widehat{H}_{[b,k]} \right)
  e^{-iu \left( \widehat{H}_{[-k,-b]} + \widehat{H}_{[b,k]} \right)} \Big] \| \label{111}
\end{align}
$ \widehat{H}_{[-k-1,-b]} + \widehat{H}_{[b,k+1]} - \widehat{H}_{[-k,-b]} - \widehat{H}_{[b,k]} $ is ``mostly'' supported around $ \pm (k+1) $ but with tails on $ [-2k,2k] $, whereas $ \widehat{H}_{[-k,k]} - \widehat{H}_{[-k,-b]} - \widehat{H}_{[b,k]} $ is ``mostly'' supported around $ [-b,b] $ with tails on $ [-2k,2k] $.
The commutator in \eqref{111} has the structure of a Lieb-Robinson bound, once we have shown that these two operators  can be approximated with operators that have supports disjoint and enough far apart to ensure a small upper bound. We now perform such approximations.

Analogously to what done in \eqref{G16} we have that:
\begin{align} 
 & \| \widehat{H}_{[-k-1,-b]} + \widehat{H}_{[b,k+1]} -   \widehat{H}_{[-k,-b]} - \widehat{H}_{[b,k]} \| \le 2J C_\xi \left(f\left(\frac{k}{\xi}\right)-f\left(\frac{2k-b}{\xi}\right) \right) +  \| \widehat{H}_{k+1} + \widehat{H}_{-(k+1)} \| \\
 & \le 2J C_\xi \left(f\left(\frac{k}{\xi}\right)-f\left(\frac{2k-b}{\xi}\right) \right) +  2J \le O(2J) \label{115}
\end{align}
\eqref{115} bounds the norm of the first entry in the commutator of \eqref{111}.


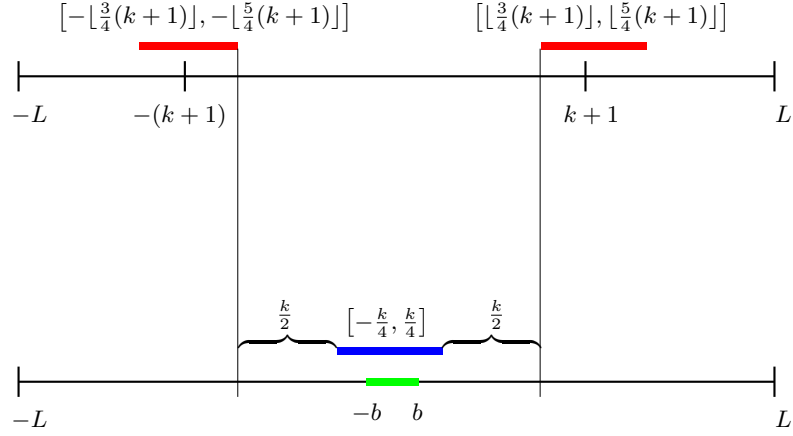
\begin{figure}[h]
\setlength{\unitlength}{1mm} 
\begin{picture}(180,0)(-50,50)

\thicklines

\put(-15,16){\line(0,1){4}}
\put(85,16){\line(0,1){4}}
\put(7,16){\line(0,1){4}}
\put(60,16){\line(0,1){4}}
\put(-15,18){\line(1,0){100}}

\linethickness{1mm}

{\color{red} \put(1,22){\line(1,0){13}}}
{\color{red} \put(53,22){\line(1,0){14}}}

\thinlines

\put(-17,12){$ -L $}
\put(84,12){$ L $}
\put(-1,12){$ -(k+1) $}
\put(56,12){$ k+1 $}
\put(44,25){$ \left[ \lfloor \frac{3}{4}(k+1)\rfloor,\lfloor \frac{5}{4}(k+1)\rfloor \right] $}
\put(-11,25){$ \left[- \lfloor \frac{3}{4}(k+1)\rfloor,-\lfloor \frac{5}{4}(k+1)\rfloor \right] $}

\end{picture}

\begin{picture}(180,90)(-50,0)

\thinlines
\linethickness{0.1mm}

\put(14,16){\line(0,1){46}}
\put(54,16){\line(0,1){46}}

\thicklines

\put(-15,16){\line(0,1){4}}
\put(85,16){\line(0,1){4}}
\put(-15,18){\line(1,0){100}}

\linethickness{1mm}

{\color{green} \put(31,18){\line(1,0){7}}}

{\color{blue} \put(26,22){\line(1,0){14}}}

\thinlines

\put(-17,12){$ -L $}
\put(84,12){$ L $}

\put(27,25){$ \left[ -\frac{k}{4},\frac{k}{4} \right] $}

\put(28,13){$ -b $}
\put(36,13){$ b $}

\put(40,22){$\overbrace{\hspace{13mm}}$}
\put(46,26){$ \frac{k}{2}$}

\put(13,22){$\overbrace{\hspace{13mm}}$}
\put(18,26){$ \frac{k}{2}$}

\end{picture}
\caption{In red the regions where the supports of the operators $ \widehat{H}_{-(k+1)} $ and $ \widehat{H}_{k+1} $ have been restricted. These operators with restricted support have been denoted $ \widehat{\widehat{H}}_{k+1} $ and $ \widehat{\widehat{H}}_{-(k+1)} $. In green the region of ``maximal'' support of the operator $ \widehat{H}_{[-k,k]} - \widehat{H}_{[-k,-b]} - \widehat{H}_{[b,k]} $, in blue the region $ \left[-\frac{k}{4},\frac{k}{4} \right] $ where the support of this operator has been restricted. This operator has been denoted $ \widehat{\widehat{H}}_{[-b,b]} $. The distance among the red and blue regions, that is the distance among the supports of $ \widehat{\widehat{H}}_{\pm (k+1)} $ and $ \widehat{\widehat{H}}_{[-b,b]} $ is approximately equal to $ \frac{k}{2} $. }
\label{fig_supp}
\end{figure}


To apply the theory of Lieb-Robinson bounds for local Hamiltonians with tails we need to ensure that the two operators that enter the commutator are far apart, since they are actually terms of the Hamiltonian, namely they have tails, we need to restrict their supports, provided the error in doing so is small. Let us consider $ \widehat{H}_{k+1} $ and $ \widehat{H}_{-(k+1)} $, that have supports $ [1,2(k+1)] $ and $ [-2(k+1),-1] $. We define:
\begin{align}
 &\widehat{\widehat{H}}_{k+1} := \widetilde{\Tr}_{[ \lfloor \frac{3(k+1)}{4} \rfloor, \lfloor \frac{5(k+1)}{4} \rfloor ]^c} H_{k+1} \\
 &\widehat{\widehat{H}}_{-(k+1)} := \widetilde{\Tr}_{[ \lfloor \frac{-5(k+1)}{4},\frac{-3(k+1)}{4} \rfloor ]^c} H_{-(k+1)}
\end{align}
$ \lfloor a \rfloor $ denotes the largest integer such that $ \lfloor a \rfloor \le a $.  With this choice the size of the support is approximately $ \frac{|k|}{2} $, see figure \ref{fig_supp}. It is easy to see that the error made is upper bounded by  $  J f\left(\frac{k}{4 \xi}\right) $, in fact:
\begin{align}
 & \widehat{H}_{k+1} -  \widetilde{\Tr}_{ [ \lfloor \frac{3(k+1)}{4} \rfloor, \lfloor \frac{5(k+1)}{4} \rfloor ]^c }  H_{k+1} = \left( \widetilde{\Tr}_{[1,2(k+1)]^c} \left( \mathds{1} - \widetilde{\Tr}_{ [ \lfloor \frac{3(k+1)}{4} \rfloor, \lfloor \frac{5(k+1)}{4} \rfloor ]^c \setminus [1,2(k+1)]^c } \right)  H_{k+1} \right) 
\end{align}
is such that: $ [ \lfloor \frac{3(k+1)}{4} \rfloor, \lfloor \frac{5(k+1)}{4} \rfloor ]^c \setminus [1,2(k+1)]^c = [1,\lfloor \frac{3(k+1)}{4} \rfloor -1 ] \, \bigcup \, [ \lfloor \frac{5(k+1)}{4} \rfloor , 2(k+1) ] $, therefore the distance of $ k+1 $ from this set is approximately $ \lfloor \frac{k}{4} \rfloor $.

On the other hand:
\begin{align}
 &\widehat{H}_{[-k,k]} - \widehat{H}_{[-k,-b]} - \widehat{H}_{[b,k]} \\
 & =  \widetilde{\Tr}_{[-2k,2k]^c} \left(\sum_{r=-k}^{- \lfloor \frac{k}{4} \rfloor-1} H_r  + \sum_{r=-\lfloor \frac{k}{4} \rfloor}^{\lfloor \frac{k}{4} \rfloor} H_r   + \sum_{r=\lfloor \frac{k}{4} \rfloor+1}^{k} H_r  \right)   \nonumber \\
 &- \widetilde{\Tr}_{[-2k,-1]^c} \left( \sum_{r=-k}^{-\lfloor\frac{k}{4}\rfloor-1} H_r + \sum_{r=-\lfloor\frac{k}{4} \rfloor}^{-b} H_r  \right)   - \widetilde{\Tr}_{[1,2k]^c} \left( \sum_{r=b}^{\lfloor \frac{k}{4} \rfloor} H_r + \sum_{r=\lfloor \frac{k}{4} \rfloor+1}^{k} H_r \right)  \label{116}
\end{align}
We pair terms from \eqref{116} as follows:
\begin{align}
 & \| \left( \widetilde{\Tr}_{[-2k,2k]^c} - \widetilde{\Tr}_{[-2k,-1]^c}  \right) \sum_{r=-k}^{-\lfloor \frac{k}{4}\rfloor-1} H_r \| \le J C_\xi \left(  f \left( \frac{k}{4 \xi} \right) - f \left( \frac{k}{\xi} \right) \right) \label{119} \\
 & \| \left( \widetilde{\Tr}_{[-2k,2k]^c} - \widetilde{\Tr}_{[1,2k]^c}  \right) \sum_{r=\lfloor \frac{k}{4} \rfloor +1}^{k} H_r \| \le J C_\xi  \left( f \left( \frac{k}{4 \xi} \right) - f \left( \frac{k}{\xi} \right) \right) \label{120}
\end{align}
It is left $ \widetilde{\Tr}_{[-2k,2k]^c} \sum_{r=-\lfloor \frac{k}{4} \rfloor}^{\lfloor \frac{k}{4} \rfloor} H_r - \widetilde{\Tr}_{[-2k,-1]^c} \sum_{r=-\lfloor \frac{k}{4} \rfloor}^{-b} H_r - \widetilde{\Tr}_{[1,2k]^c} \sum_{r=b}^{\lfloor \frac{k}{4} \rfloor} H_r $. We replace this operator with the same one but with support restricted to $ [-\lfloor \frac{k}{2} \rfloor, \lfloor \frac{k}{2} \rfloor] $, we denote this operator as:
\begin{align}
 \widehat{\widehat{H}}_{[-b,b]} : = \widetilde{\Tr}_{[-\lfloor \frac{k}{2} \rfloor,\lfloor \frac{k}{2} \rfloor]^c} \sum_{r=- \lfloor \frac{k}{4} \rfloor}^{\lfloor \frac{k}{4} \rfloor} H_r - \widetilde{\Tr}_{[-\lfloor \frac{k}{2} \rfloor,-1]^c} \sum_{r=-\lfloor \frac{k}{4} \rfloor}^{-b} H_r - \widetilde{\Tr}_{[1,\lfloor \frac{k}{2} \rfloor]^c} \sum_{r=b}^{\lfloor \frac{k}{4} \rfloor} H_r
\end{align}
As shown in \eqref{221}, it follows that: 
\begin{align}
 & \| \widetilde{\Tr}_{[-2k,2k]^c} \sum_{r=-\lfloor \frac{k}{4} \rfloor}^{\lfloor \frac{k}{4} \rfloor} H_r - \widetilde{\Tr}_{[-2k,-1]^c} \sum_{r=- \lfloor \frac{k}{4} \rfloor}^{-b} H_r - \widetilde{\Tr}_{[1,2k]^c} \sum_{r=b}^{\lfloor \frac{k}{4} \rfloor} H_r -\widehat{\widehat{H}}_{[-b,b]} \| \label{221} \\
& =  \| \widetilde{\Tr}_{[-2k,2k]^c} \left( \mathds{1} - \widetilde{\Tr}_{[- \lfloor \frac{k}{2} \rfloor,\lfloor \frac{k}{2} \rfloor]^c \setminus [-2k,2k]^c } \right) \sum_{r=-\lfloor \frac{k}{4} \rfloor}^{\lfloor \frac{k}{4} \rfloor} H_r \nonumber \\
& - \widetilde{\Tr}_{[-2k,-1]^c} \left( \mathds{1} - \widetilde{\Tr}_{[-\lfloor \frac{k}{2} \rfloor,-1]^c \setminus [-2k,-1]^c } \right) \sum_{r=-\lfloor \frac{k}{4} \rfloor}^{-b} H_r \nonumber \\
&- \widetilde{\Tr}_{[1,2k]^c} \left( \mathds{1} - \widetilde{\Tr}_{[1,\lfloor \frac{k}{2} \rfloor]^c \setminus  [1,2k]^c } \right) \sum_{r=b}^{\lfloor \frac{k}{4} \rfloor} H_r \| \\
& \le J C_\xi \left( 4f \left( \frac{k}{4 \xi} \right) - 2f \left( \frac{k}{2 \xi} \right) - 2f \left( \frac{\frac{k}{2} - b }{\xi} \right) \right) \label{122}
\end{align}

To upper bound the norm of the commutator (Lieb-Robinson bound) in \eqref{111}, we also need to upper bound the norm of each entry of the commutator. The norm of the first entry is provided by \eqref{115}.

The norm of the second entry of the commutator in \eqref{111} is bounded as follows:
\begin{align}
 & \| \widehat{H}_{[-k,k]} - \widehat{H}_{[-k,-b]} - \widehat{H}_{[b,k]} \| := 
 \| \widetilde{\Tr}_{[-2k,2k]^c} \sum_{r=-k}^k H_r - \widetilde{\Tr}_{[-2k,-1]^c} \sum_{r=-k}^{-b} H_r
 - \widetilde{\Tr}_{[1,2k]^c} \sum_{r=b}^k H_r \| \\
 & \le | \widetilde{\Tr}_{[-2k,2k]^c} \left( \mathds{1} - \widetilde{\Tr}_{ [-2k,-1]^c \setminus [-2k,2k]^c } \right) \sum_{r=-k}^{-b} H_r \| + \| \widetilde{\Tr}_{[-2k,2k]^c} \left( \mathds{1} - \widetilde{\Tr}_{ [1,2k]^c \setminus [-2k,2k]^c } \right) \sum_{r=b}^{k} H_r \| +  \| \widehat{H}_{[-b+1,b-1]} \| \\
& \le |  \left( \mathds{1} - \widetilde{\Tr}_{ [-2k,-1]^c \setminus [-2k,2k]^c } \right) \sum_{r=-k}^{-b} H_r \| + \| \left( \mathds{1} - \widetilde{\Tr}_{ [1,2k]^c \setminus [-2k,2k]^c } \right) \sum_{r=b}^{k} H_r \| +  (2b-1)J \\ 
& \le 2J \sum_{r=b}^k f\left( \frac{r}{\xi} \right) +  (2b-1)J \le 2JC_\xi \left(f\left( \frac{b}{\xi} \right) - f\left( \frac{k}{\xi} \right) \right) +  (2b-1)J \le O(2bJ)  \label{128}
\end{align}

Overall the norm of the commutator in \eqref{111} is bounded as follows:
\begin{align}
 &\| \Big[ \widehat{H}_{[-k-1,-b]} + \widehat{H}_{[b,k+1]} - 
  \widehat{H}_{[-k,-b]} - \widehat{H}_{[b,k]} ,  e^{iu \left( \widehat{H}_{[-k,-b]} + \widehat{H}_{[b,k]} \right)}
  \left( \widehat{H}_{[-k,k]} - \widehat{H}_{[-k,-b]} - \widehat{H}_{[b,k]} \right)
  e^{-iu \left( \widehat{H}_{[-k,-b]} + \widehat{H}_{[b,k]} \right)} \Big] \| \nonumber \\
  & \le \| \Big[ \widehat{\widehat{H}}_{k+1} + \widehat{\widehat{H}}_{-(k+1)} ,  e^{iu \left( \widehat{H}_{[-k,-b]} + \widehat{H}_{[b,k]} \right)}
  \left( \widehat{H}_{[-k,k]} - \widehat{H}_{[-k,-b]} - \widehat{H}_{[b,k]} \right)
  e^{-iu \left( \widehat{H}_{[-k,-b]} + \widehat{H}_{[b,k]} \right)} \Big] \| + \nonumber \\ 
  & + 2 \left( 2J C_\xi \left(f\left(\frac{k}{\xi}\right)-f\left(\frac{2k-b}{\xi}\right) \right) + 2 J  f\left(\frac{k}{4\xi}\right) \right) \left( 2JC_\xi \left(f\left( \frac{b}{\xi} \right) - f\left( \frac{k}{\xi} \right) \right) +  (2b-1)J  \right)  \\
  & \le \| \Big[ \widehat{\widehat{H}}_{k+1} + \widehat{\widehat{H}}_{-(k+1)} ,  e^{iu \left( \widehat{H}_{[-k,-b]} + \widehat{H}_{[b,k]} \right)}   \widehat{\widehat{H}}_{[-b,b]} 
  e^{-iu \left( \widehat{H}_{[-k,-b]} + \widehat{H}_{[b,k]} \right)} \Big] \| + \label{norm_comm} \\
  & + 8\,O(bJ^2) \left( 2 C_\xi \left(f\left(\frac{k}{\xi}\right)-f\left(\frac{2k-b}{\xi}\right) \right) + 2  f\left(\frac{k}{4\xi}\right) \right)  + 4\,O(C_\xi J^2)   \left( 6f \left( \frac{k}{4 \xi} \right) - 4f \left( \frac{k}{2 \xi} \right) - 2f \left( \frac{\frac{k}{2} - b }{\xi} \right) \right)  \nonumber \\
  & \le O(bJ^2)(e^{\vlr t}-1) f\left( \frac{k}{2\xi} \right) + J^2 D_\xi \label{def_D}
\end{align}
In the last line, \eqref{def_D}, we have introduced $ D_\xi \le O \left(f\left( \frac{k}{4\xi} \right) \right) $ to shorten the notation. We see that in comparison to the Lieb-Robinson bound of strictly local operators there is a term, $ D_\xi $ that vanishes for large distances, of $ O(k) $, among the supports of $ \widehat{\widehat{H}}_{k+1} $ and $ \widehat{\widehat{H}}_{[-b,b]} $, as $ f\left( \frac{k}{2\xi} \right) $. We stress that when $ f $ has a compact support, as discussed before in \eqref{G16}, we recover the case of strictly local Hamiltonians and $ D_\xi $ is identically vanishing. 
The norm of the commutator in \eqref{norm_comm} is upper bounded by a Lieb-Robinson bound for Hamiltonians with rapidly decaying interactions, see for example equation S17 of \cite{Yin_2023}, we denote $ \vlr $ the corresponding velocity. For a one dimensional system the L-R bound in equation S17 of \cite{Yin_2023} is independent of the size of the supports of the operators.  We notice that the supports of the operators involved in the Lieb-Robinson bound are at a distance of the order of $ \frac{k}{2} $.

We are now ready to go back to \eqref{99}, that we copy below, to obtain the bound on the trace distance that we are looking for:
\begin{align}
 & \frac{1}{2}\|  {\widehat V}_{\Lambda_{k+1}}(t) \rho {\widehat V}^*_{\Lambda_{k+1}}(t) -  {\widehat V}_{\Lambda_{k}}(t) \rho {\widehat V}^*_{\Lambda_{k}}(t) \|_1  \\
 & \le  t \sup_{s\in[0,t]} \|  e^{is \left( \widehat{H}_{[-k,-b]} + \widehat{H}_{[b,k]} \right)} \left( \widehat{H}_{[-k,k]} - \widehat{H}_{[-k,-b]} - \widehat{H}_{[b,k]} \right) e^{-is \left( \widehat{H}_{[-k,-b]} + \widehat{H}_{[b,k]} \right)} +  \nonumber \\
  & - e^{is \left( \widehat{H}_{[-k-1,-b]} + \widehat{H}_{[b,k+1]} \right)} \left( \widehat{H}_{[-k,k]} - \widehat{H}_{[-k,-b]} - \widehat{H}_{[b,k]} \right) e^{-is \left( \widehat{H}_{[-k-1,-b]} + \widehat{H}_{[b,k+1]} \right)}  \| +  \nonumber \\
  & +  t \, \|  \widehat{H}_{[-k-1,k+1]} - \widehat{H}_{[-k-1,-b]} - \widehat{H}_{[b,k+1]}  -  \left( \widehat{H}_{[-k,k]} - \widehat{H}_{[-k,-b]} - \widehat{H}_{[b,k]} \right) \| \\
 &  \le  \, t \sup_{s\in [0,t]} \int_0^s du  \| \Big[ \widehat{H}_{[-k-1,-b]} + \widehat{H}_{[b,k+1]} - 
  \widehat{H}_{[-k,-b]} - \widehat{H}_{[b,k]} ,  \nonumber \\ 
  & e^{iu \left( \widehat{H}_{[-k,-b]} + \widehat{H}_{[b,k]} \right)}
  \left( \widehat{H}_{[-k,k]} - \widehat{H}_{[-k,-b]} - \widehat{H}_{[b,k]} \right)
  e^{-iu \left( \widehat{H}_{[-k,-b]} + \widehat{H}_{[b,k]} \right)} \Big] \| + \nonumber \\
  & + 4 \, t \,  J \, C_\xi \left( f\left(\frac{k}{\xi}\right) - f\left(\frac{2k}{\xi}\right) \right) + t \,J \, f\left(\frac{k+1}{\xi}\right)  \label{125} \\
  & \le  t^2 O(bJ^2) (e^{\vlr t}-1) f\left( \frac{k}{2\xi} \right) + t^2J^2 D_\xi   +  t \,  J \, D'_\xi \label{126} \\ &\le t \, J \, O \left(  g(b,J,t) \, f \left( \frac{k}{4 \xi} \right) \right) \label{127}
\end{align}
In \eqref{126} we have introduced 
\begin{align}
 D'_\xi:= 4  \, C_\xi \left( f\left(\frac{k}{\xi}\right) - f\left(\frac{2k}{\xi}\right) \right) + f\left(\frac{k+1}{\xi}\right) \le O \left(f\left( \frac{k}{\xi} \right) \right)
\end{align}
When $ f $ has a compact support, see equation \eqref{f_compact}, $ D'_\xi $, as $ D_\xi $, vanishes.
In \eqref{127}, it is: 
\begin{align} \label{g_def}
g(b,J,t) \le b \,t\, J \,\left( (e^{\vlr t}-1) + O(1) \right) + O(1)
\end{align}
The minimization with respect to $ b $ of the trace distance \eqref{127} is straightforward, in fact being $ g $ proportional to $ b $ in \eqref{g_def} up to corrections of $ O(1) $, the minimum is reached for the smallest $ b $ that makes the theory consistent. Since we have seen for the strictly local case, see \eqref{f_compact}, that $ b=\xi $, then we pick $ b=\xi $ also in the general case.

We are now ready to upper bound the sum in \eqref{bound_sum_tails}, that we copy below \eqref{bound_sum_tails_app}, to obtain $ \Delta S_\alpha(t) $ of equation \eqref{bound_Rényi_tails} in the case of an Hamiltonian with exponential decrease of interactions.
As done in the case of strictly local Hamiltonians we need to find out the value of $ l $ that minimizes the sum. This will turn out to be dependent on $ \alpha $, $ v_{LR} $ and $ \xi $. We adopt a different, but equivalent, approach than the one of appendix \ref{sum_details} where we had an approximate value of $ l $ from the solution of the minimization condition. Here we will pick up $ l=\beta t $ with a $ \beta $ to be determined in such a way to hold a linear increase in $ t $ for large $ t $ of $ \Delta S_\alpha(t) $, we know in fact that a ``slower'' increase has been already ruled out by \cite{Schuch_Wolf_2008}. 
With, for large enough $ t $, $ R_k(t) \le \xi (t J)^2 e^{\vlr t}e^{-\frac{k}{4\xi}} $ and 
\begin{align} \label{bound_sum_tails_app}
  \Delta S_\alpha (t) & \le \frac{1}{1-\alpha} \sum_{k=l}^{L-1} \log \left[ \left(1-R_k(t) \right)^\alpha + (2^{2(k+1)+1}-1)^{1-\alpha}R_k(t)^\alpha \right] + 2l +1 \\
&\le \frac{1}{1-\alpha} \sum_{k=l}^{\infty}  \left( -\alpha R_k(t)  + 2^{(2k +3)(1-\alpha)}R_k(t)^\alpha \right) + 2\beta t +1 \\
&\le \frac{1}{1-\alpha} \sum_{n=0}^{\infty}  \left( -\alpha R_{l+n}(t)  + 2^{(2(l+n) +3)(1-\alpha)}R_{l+n}(t)^\alpha \right) + 2\beta t +1 \\
&\le \frac{1}{1-\alpha} \sum_{n=0}^{\infty}  \left( -\alpha \xi (t J)^2 e^{\vlr t}e^{-\frac{\beta t+n}{4\xi}}  + e^{\ln 2(2(\beta t+n) +3)(1-\alpha)}\left(\xi (t J)^2 \right)^\alpha e^{\alpha \vlr t}e^{-\alpha\frac{\beta t+n}{4\xi}} \right) + 2\beta t +1 \\
&\le \frac{1}{1-\alpha} \sum_{n=0}^{\infty}  \left( -\alpha \xi (t J)^2 e^{\vlr t}e^{-\frac{\beta t+n}{4\xi}}  + e^{3\ln 2(1-\alpha)}\left(\xi (t J)^2 \right)^\alpha e^{\alpha \vlr t}e^{-\lambda(\beta t+n)} \right) + 2\beta t +1 \label{rest_sum}
\end{align} 
In the last step we have introduced $ \lambda:= \frac{\alpha}{4\xi} - 2\ln2(1-\alpha) $. We now carry out two important steps. The first is to ensure that the sum over $ n $ is finite, this gives a condition on the set of allowed $ \alpha$. The second step is to impose that the sum of the series decreases exponentially in time, this will provide a restriction on the value of $ \beta $ and therefore the large $ t $ behaviour of $ \Delta S_\alpha (t) $.
From the second sum in \eqref{rest_sum} we get a restriction on the value of $ \alpha $ as follows:
\begin{align} \label{alpha_restr}
 \lambda > 0 \hspace{5mm} \Rightarrow  \hspace{5mm} \frac{\alpha}{4 \xi} - 2\ln2(1-\alpha) > 0 \hspace{5mm} \Rightarrow  \hspace{5mm} \alpha > \frac{1}{1+\frac{1}{8\xi\ln2}}
\end{align}
With $ \xi \ll 1 $, that corresponds to the Hamiltonians' term becoming almost on site, \eqref{alpha_restr} reduces to $ \alpha \gtrsim 0 $. In the limit of extendend interactions instead $ \xi \gg 1 $, being by assumption $0 < \alpha \le 1$, \eqref{alpha_restr} gives $ \alpha \approx 1 $.

On the other hand the restriction on $ \beta $ to ensure that the sum in \eqref{rest_sum} decreases exponentially in $ t $ provides:
\begin{align} \label{beta_restr}
 \lambda \beta - \alpha \vlr > 0 \hspace{5mm} \Rightarrow  \hspace{5mm} \beta > \frac{\alpha \vlr}{\lambda} > \frac{v}{\frac{1}{4\xi}- (2\ln2)\frac{1-\alpha}{\alpha}}
\end{align}
It can be checked that the condition for the exponential decrease in time of the first term of the sum in \eqref{rest_sum} is less restrictive than \eqref{beta_restr}.

Overall, making use of the conditions \eqref{alpha_restr} and \eqref{beta_restr} we have that:
\begin{align} \label{final_bound_tails_app}
& \Delta S_\alpha (t)  \le \frac{1}{1-\alpha} \Big[ -\alpha \xi (t J)^2 e^{\vlr t}e^{-\frac{\beta t}{4\xi}} \frac{e^{\frac{1}{4\xi}}}{e^{\frac{1}{4\xi}}-1} +  e^{(3\ln 2)(1-\alpha)}\left(\xi (t J)^2 \right)^\alpha e^{\alpha \vlr t}e^{-\lambda \beta t} \frac{e^\lambda}{e^\lambda -1} \Big] +1  + 2\beta t 
\end{align}
Considering the lower bound on $ \beta $ \eqref{beta_restr}, the best upper bound on $ \Delta S_\alpha (t) $ reads
\begin{align} \label{final_rate_tails}
& \lim_{t\rightarrow \infty} \lim_{L\rightarrow \infty}  \frac{\Delta S_\alpha (t)}{t} 
\le  \frac{c'\vlr}{\frac{1}{4\xi}- (2\ln2)\frac{1-\alpha}{\alpha}} 
\end{align}
with $ c' > 2 $.

The limit $ \alpha \rightarrow 1 $ gives the von Neumann entropy. The factor proportional to $ \frac{1}{1-\alpha} $ in \eqref{final_bound_tails_app} goes to zero  with $ \alpha \rightarrow 1 $, therefore being differentiable in $ \alpha $, the limit $ \alpha \rightarrow 1 $, equals to minus the derivative of such factor in $ \alpha=1 $. It is immediate to realize that such derivative is exponentially decreasing in $ t $, then it follows:
\begin{align} \label{G66}
  \lim_{t\rightarrow \infty} \lim_{L\rightarrow \infty} \frac{\Delta S_1 (t)}{t} \le 4c'\vlr \xi
\end{align}
This shows that the limits $ \alpha \rightarrow 1 $ and the double limit $ \lim_{t\rightarrow \infty} \lim_{L\rightarrow \infty} $ of the rate $\frac{\Delta S_1 (t)}{t} $ can be exchanged leading to the same upper bound in \eqref{G66}, that in fact coincides with the limit $ \alpha \rightarrow 1 $ of the rhs of \eqref{final_rate_tails}.


\section{An upper bound on the concavity of the $\alpha$-Rényi entropies, $ 0 < \alpha < 1 $}  \label{Up_Renyi_Con}

\begin{Lemma} \label{Renyi_up}
Given a convex combination of states $ \rho := \sum_{i=1}^N p_i \rho_i $, with $ \rho_i: \mathds{C}^d \rightarrow \mathds{C}^d $, denoting $ \rho_i= \sum_{j=1}^{n_i} \lambda_j^i |e_j^i \rangle \langle e_j^i | $ the spectral  decomposition of each $ \rho_i $, for $ 0 < \alpha < 1 $, it holds:
\begin{align} \label{Rényi_upper_conc}
 \sum_{i=1}^N p_i S_\alpha (\rho_i) \le S_\alpha (\rho) \le H_\alpha \{p_i \lambda_j^i \} 
\end{align}
where $ H_\alpha \{p_i \lambda_j^i \} := \frac{1}{1-\alpha}  \log \sum_{i=1}^N \sum_{j=1}^{n_i} ( p_i \lambda_j^i)^\alpha $ is the $ \alpha$-Rényi entropy of the probability distribution $ \{p_i \lambda_j^i \} $. Moreover:
 \begin{equation} \label{up_con_sup}
  H_\alpha \{p_i \lambda_j^i \} \le H_\alpha \{p_i\} + \max_i \{ S_\alpha(\rho_i) \} 
\end{equation}
Therefore if $ \{ \rho_i \} $ are pure states, it is $ S_\alpha (\rho) \le H_\alpha \{p_i \} $, with the upper bound depending only on the probability distribution $ \{ p_i \} $.
\end{Lemma}

\begin{proof}
We now sketch a proof, for more details consult lemma 1.24 (page 18) and example 11.12 (page 178) of \cite{Petz_Matrix_Analysis}.
The first inequality in \eqref{Rényi_upper_conc} is the concavity of the Rényi entropy with $ 0 < \alpha <1 $, \cite{Rastegin_2011,Vershynina_Scholarpedia}.
 The proof of the second inequality in \eqref{Rényi_upper_conc} is a simple application of the theory of majorization, see chapter II of \cite{Bhatia}. The spectral decomposition of $ \rho $ reads: $ \rho = \sum_{k=1}^m a_k |v_k \rangle \langle v_k | $. Since $ \{|v_k \rangle \} $ is an orthonormal set, it follows that $ m \le d $. The set $ \{ |e_j^i \rangle \} $, that collects the eigenvectors of all $ \{ \rho_i \} $, is in general not orthonormal. In fact, for example, two states $ \rho_m $ and $ \rho_n $  can share an eigenvector.  It follows that, completing the matrix $ V $ with entries $ V_{j,k}^i := \sqrt{\frac{p_i \lambda_j^i}{a_k}} \langle v_k |e_j^i \rangle $ to a unitary matrix $ U $, it holds
 \begin{equation}
  |e_j^i \rangle = \sum_k U_{j,k}^i |v_k \rangle 
 \end{equation}
It follows from the definition of $ V $, and the unitarity of $ U $, that 
\begin{equation}
 p_i \lambda_j^i = \sum_k |U_{j,k}^i|^2 a_k
\end{equation}
The matrix with entries $ |U_{j,k}^i|^2 $ is a double stochastic matrix, meaning that $ \sum_k |U_{j,k}^i|^2 = \sum_{i,j} |U_{j,k}^i|^2 = 1 $, this implies, see theorem II.1.10 (page 33) of \cite{Bhatia},   that the set $ \{p_i \lambda_j^i \} $ is majorized by the set  $ \{ a_k \} $. Each set is supposed to be ordered in a non increasing fashion. It then follows from theorem II.3.1 (page 40) of \cite{Bhatia}, and the concavity of the function $ x^\alpha $, with $0< \alpha <1 $, that:
\begin{equation} \label{x_concave} 
 \sum_k a_k^\alpha \le \sum_{i,j} (p_i \lambda_j^i)^\alpha 
\end{equation}
That implies, being the logarithm an increasing function: $ S_\alpha (\rho) := H_\alpha \{a_k \} \le H_\alpha \{p_i \lambda_j^i \}$.

We now prove equation \eqref{up_con_sup}.
 \begin{align}
 & H_\alpha \{p_i \lambda_j^i \} := \frac{1}{1-\alpha}  \log \sum_{i,j} ( p_i \lambda_j^i)^\alpha \nonumber \\ & =
 \frac{1}{1-\alpha}  \log \sum_{i=1}^N  p_i^\alpha \sum_j (\lambda_j^i)^\alpha \\ 
 &= \frac{1}{1-\alpha}  \log \sum_{i=1}^N  p_i^\alpha \exp \left( (1-\alpha) \frac{1}{1-\alpha} \log \sum_j (\lambda_j^i)^\alpha \right) \nonumber \\
 & = \frac{1}{1-\alpha}  \log \sum_{i=1}^N  p_i^\alpha \exp \left( (1-\alpha) S_\alpha (\rho_i) \right) \\
 & \le \frac{1}{1-\alpha}  \log \sum_{i=1}^N   p_i^\alpha \exp \left( (1-\alpha) \max_i \{ S_\alpha(\rho_i) \}  \right) \nonumber \\ & =  H_\alpha \{p_i\} + \max_i \{ S_\alpha(\rho_i) \}
\end{align}
\end{proof}

Equation \eqref{Rényi_upper_conc} holds also for the von Neumann entropy, with the Rényi entropy $ H_\alpha $ replaced by the Shannon entropy $ H $. In fact the same proof applies with the function $ x^\alpha $ replaced by $ -x \log x $ in equation \eqref{x_concave}. See equation 2.3 of \cite{Wehrl_Review}, and also theorem 11.10 (page 518) of \cite{Nielsen_Chuang}, or theorem 3.7 (page 35) of \cite{Petz_Quantum_Information_book} for alternative proofs.  
For the von Neumann entropy we can see that the upper bound takes a more explicit form:
 \begin{equation} \label{NC}
  S(\rho) \le H\{p_i\} + \sum_i p_i S(\rho_i)  
 \end{equation}
In fact:
\begin{align}
  H \{p_i \lambda_j^i \}&:= -\sum_{i,j} p_i \lambda_j^i \log ( p_i \lambda_j^i ) \nonumber \\ &= -\sum_{i,j} p_i \lambda_j^i \left( \log p_i + \log \lambda_j^i \right) \\
 & = -\sum_{i} p_i \log p_i  -\sum_{i,j} p_i \lambda_j^i \log \lambda_j^i \nonumber \\& = H\{p_i\} + \sum_i p_i S(\rho_i) 
\end{align}
This implies that the difference $ S(\rho) - \sum_i p_i S(\rho_i) \le  H\{p_i\} $ is upper bounded by a quantity only depending on the probability distribution $ \{ p_i \} $, but not on other quantities like the states $ \{ \rho_i \} $ or the dimension of the Hilbert space. The bound \eqref{NC} has been improved in theorem 14 of \cite{Audenaert_2014}, where it also appears an upper bound on the difference of Holevo quantities for different ensembles that is independent from the Hilbert space dimension.

In general this is not the case for the Rényi entropy. Let us consider the state $ \rho $ of equation \eqref{saturating_states} and the maximally mixed state of $ \mathds{C}^d $, $ \frac{1}{d} \mathds{1}_d $. It is easy to see that $ S_\alpha \left( p \rho + (1-p) \frac{1}{d} \mathds{1}_d \right) \approx  \log d $, with fixed $ \alpha $ and $ p $, both not too close to $1$, and $ d \gg 1 $. Then: $ S_\alpha \left( p \rho + (1-p) \frac{1}{d} \mathds{1}_d \right) -(1-p) S_\alpha \left( \frac{1}{d} \mathds{1}_d \right) \approx p \log d $, depending on the size of the Hilbert space.

Within the setting of the physical system described in section \ref{sec_Physical_Setting}, and the notations of Lemma \ref{Renyi_up}, let us consider the case where $ \max \{ S_\alpha(\rho_i) \} \le O( \log L ) $, and $ N \le O(L) $, then according to equation \eqref{up_con_sup}, we have:
\begin{align} \label{convex_comb}
 S_\alpha(\rho) \le \frac{1}{1-\alpha} \log \sum_{i=1}^N p_i^\alpha + O(\log L ) \le O(\log L)
\end{align}
This shows that when convex combinations of few, $ O(L) $, low entangled states, $ S_\alpha(\rho_i) \le O(\log L) $, are considered then $ \rho = \sum_i p_i \rho_i $ is still low entangled: $ S_\alpha(\rho) \le O(\log L) $.

\section{Proof of equation \eqref{Renyi-logneg}: the $\frac{1}{2}$-Rényi entropy (of the partial trace) equals the logarithmic negativity for pure states} \label{negativity_renyi}

In this short section we reproduce the argument in proposition 8 of \cite{Vidal_Werner}, see also the appendix in \cite{Nach_2013}, leading to the proof that the $\frac{1}{2}$-Rényi entropy equals the logarithmic negativity for pure states.

Given a state $ \rho $ defined on the tensor product of finite dimensional Hilbert spaces $ \mathcal{H}_1 \otimes \mathcal{H}_2 $, the partial transpose with respect to the first space $ \rho^{T_1} $ has matrix elements such that: 
\begin{align}
 \langle e_j^1 \otimes e_k^2, \rho^{T_1} (e_l^1 \otimes e_m^2) \rangle = \langle e_l^1 \otimes e_k^2, \rho (e_j^1 \otimes e_m^2) \rangle
\end{align}
The characterization of separable states by the properties of their partial transpose started with the Peres criterium \cite{Peres_1996}, stating that a necessary condition for separability is that all the eigenvalues of the partial transpose are negative. This is also sufficient, as established in \cite{Horo_2001}, see their Remark 2.
Vidal and Werner defined the logarithmic negativity in \cite{Vidal_Werner} as:
\begin{align}
 E_\mathcal{N}(\rho) := \log_2 \| \rho^{T_1} \|_1
\end{align}
Following proposition 8 of \cite{Vidal_Werner}, let us show that for a pure state $ \rho = | \Phi \rangle \langle \Phi | $, and the Schmidt decomposition of $ | \Phi \rangle = \sum_\alpha c_\alpha | e_\alpha^1 \otimes e_\alpha^2 \rangle$, with $ c_\alpha > 0 $, it holds:
\begin{equation} \label{Renyi-logneg_app}
 S_{1/2}(\Tr_{\mathcal{H}_1} \rho) := 2 \log \sum_\alpha c_\alpha = \log_2 \| \rho^{T_1} \|_1
\end{equation}

It is well known that the common non zero eigenvalues of $ \Tr_{\mathcal{H}_1} \rho $ and $ \Tr_{\mathcal{H}_2} \rho $ is the set $ \{ c_\alpha^2 \} $. 
\begin{align}
 ( | \Phi \rangle \langle \Phi | )^{T_1} = \sum_{\alpha,\beta} c_\alpha c_\beta ( | e_\alpha^1 \otimes e_\alpha^2 \rangle \langle e_\beta^1 \otimes e_\beta^2 | )^{T_1}
\end{align}
It is: 
\begin{align}
( | e_\alpha^1 \otimes e_\alpha^2 \rangle \langle e_\beta^1 \otimes e_\beta^2 | )^{T_1} =  | e_\beta^1 \otimes e_\alpha^2 \rangle \langle e_\alpha^1 \otimes e_\beta^2 | 
\end{align}
in fact
\begin{align}
 \langle e_\gamma^1 \otimes e_\delta^2 | e_\beta^1 \otimes e_\alpha^2 \rangle \langle e_\alpha^1 \otimes e_\beta^2 | e_\lambda^1 \otimes e_\eta^2 \rangle = \delta_{\gamma,\beta} \delta_{\delta,\alpha} \delta_{\alpha,\lambda} \delta_{\beta,\eta}  
\end{align}
on the other hand
\begin{align}
 \langle e_\lambda^1 \otimes e_\delta^2 | e_\alpha^1 \otimes e_\alpha^2 \rangle \langle e_\beta^1 \otimes e_\beta^2 | e_\gamma^1 \otimes e_\eta^2 \rangle = \delta_{\lambda,\alpha} \delta_{\delta,\alpha} \delta_{\beta,\gamma} \delta_{\beta,\eta}  
\end{align}
Then:
\begin{align} \label{partial_t}
 ( | \Phi \rangle \langle \Phi | )^{T_1} = \sum_{\alpha,\beta} c_\alpha c_\beta | e_\beta^1 \otimes e_\alpha^2 \rangle \langle e_\alpha^1 \otimes e_\beta^2 | 
\end{align}
Let us define the operator $ F : \mathcal{H}_1 \otimes \mathcal{H}_2 \rightarrow \mathcal{H}_1 \otimes \mathcal{H}_2 $, such that $ F(e_\alpha^1 \otimes e_\beta^2) = (e_\beta^1 \otimes e_\alpha^2) $. It is immediate to verify that $ F $ is unitary, in fact $ F^{-1} $ coincides with $ F $, while the adjoint, $ F^* $, of $ F $ is such that given any two vectors $ \psi $ and $ \phi $ in  $  \mathcal{H}_1 \otimes \mathcal{H}_2 $
\begin{align}
 & \psi = \sum_{\alpha, \beta} c_\alpha c_\beta | e_\alpha^1 \otimes e_\beta^2 \rangle,  \hspace{5mm} \phi = \sum_{\gamma, \delta} d_\gamma d_\delta | e_\gamma^1 \otimes e_\delta^2 \rangle 
\end{align}
\begin{align}
\langle \phi | F^* | \psi \rangle & := \overline{ \langle \psi | F | \phi \rangle } \\ 
& = \sum_{\alpha, \beta, \gamma, \delta} c_\alpha c_\beta \overline{d_\gamma d_\delta} \langle e_\alpha^1 \otimes e_\beta^2 | F | e_\gamma^1 \otimes e_\delta^2 \rangle \\
& = \left( \sum_\alpha c_\alpha \bar{d_\alpha} \right)^2  \label{adjoint_elements}
\end{align} 
It can be easily checked that $ \langle \phi | F | \psi \rangle $ coincides with \eqref{adjoint_elements}. 
Defining
\begin{align}
 C_1 := \sum_\alpha c_\alpha | e_\alpha^1 \rangle \langle e_\alpha^1 |, \hspace{1cm} C_2 := \sum_\beta c_\beta | e_\beta^2 \rangle \langle e_\beta^2 |
\end{align}
$ F $ as defined above is extended $ F : \mathcal{B}(\mathcal{H}_1 \otimes \mathcal{H}_2) \rightarrow \mathcal{B}(\mathcal{H}_1 \otimes \mathcal{H}_2) $, in such a way it acts like the identity on ``bras''. This implies that $ F(C_1 \otimes C_2) $ equals the RHS of \eqref{partial_t}:
\begin{align}
 F(C_1 \otimes C_2) := \sum_{\alpha,\beta} c_\alpha c_\beta | e_\beta^1 \otimes e_\alpha^2 \rangle \langle e_\alpha^1 \otimes e_\beta^2 |
\end{align}
This extended $ F $ is unitary as well. This follows from the proof above or it can be checked directly over the space of matrices with complex entries   $ \mathcal{B}(\mathcal{H}_1 \otimes \mathcal{H}_2) $, where the adjoint $ F^* $ is this time defined, with $ A, B \in \mathcal{B}(\mathcal{H}_1 \otimes \mathcal{H}_2) $, as:
\begin{align}
 \langle A, F^* (B) \rangle = \overline{\langle B, F (A) \rangle} = \overline{\Tr( B^* F (A))}
\end{align}
Finally we have that:
\begin{align}
 & \| ( | \Phi \rangle \langle \Phi | )^{T_1} \|_1 = \| F(C_1 \otimes C_2) \|_1 = \| C_1 \otimes C_2 \|_1 \\
 & = \| C_1 \|_1 \| C_2 \|_1 = \left( \sum_\alpha c_\alpha \right)^2  
\end{align}
that proves equation \eqref{Renyi-logneg_app}.



%



%

%


\bibliography{bibliography_Renyi}

\end{document}